\documentclass[12pt]{article}
\usepackage{amsmath,amsthm,amscd,amssymb}
\usepackage{tikz}
\usepackage{graphicx,psfrag,epsf}
\usepackage{enumerate}
\usepackage{natbib}

\usepackage{psfig}
\usepackage{epsfig}
\usepackage{epstopdf}

\usepackage{multirow}
\usepackage{booktabs}
\usepackage{blkarray}

\usepackage{hyperref}

\usepackage{url}
\usepackage{pdfsync}
\usepackage{blkarray}

\usepackage{algorithmic}
\usepackage{algorithm}

\usepackage{colortbl}

\newtheorem{prop}{Proposition}

\newtheorem{dfn}{Definition}

\newcommand{\eps}{\epsilon}
\newcommand{\pa}{\partial}

\renewcommand{\eps}{\varepsilon}
\renewcommand{\epsilon}{\varepsilon}
\renewcommand{\Sigma}{\varSigma}

\title{Lagrangian Dynamical Monte Carlo} 

\author{Shiwei Lan\footnote{Department of Statistics, University of California, Irvine, USA.}, \quad Vassilios Stathopoulos\footnote{Department of Statistical Science, University College London, UK.}, \quad Babak Shahbaba$^{*}$, \quad Mark Girolami$^{\dagger}$ } 

\begin{document}
\maketitle

\bigskip
\begin{abstract}
Hamiltonian Monte Carlo (HMC) improves the computational efficiency of the Metropolis algorithm by reducing its random walk behavior. Riemannian Manifold HMC (RMHMC) further improves HMC's performance by exploiting the geometric properties of the parameter space. However, the geometric integrator used for RMHMC involves implicit equations that require costly numerical analysis (e.g., fixed-point iteration). In some cases, the computational overhead for solving implicit equations undermines RMHMC's benefits. To avoid this problem, we propose an explicit geometric integrator that replaces the momentum variable in RMHMC by velocity. We show that the resulting transformation is equivalent to transforming Riemannian Hamilton dynamics to Lagrangian dynamics. Experimental results show that our method improves RMHMC's overall computational efficiency. All computer programs and data sets are available online (\url{http://www.ics.uci.edu/~babaks/Site/Codes.html}) in order to allow replications of the results reported in this paper. 
\end{abstract}

\section{Introduction}

Hamiltonian Monte Carlo (HMC) \citep{duane87} reduces the random walk behavior of Metropolis by proposing samples that are distant from the current state, but nevertheless have a high probability of acceptance. These distant proposals are found by numerically simulating Hamiltonian dynamics for some specified amount of fictitious time \citep{neal10}. Hamiltonian dynamics can be represented by a function, known as the \emph{Hamiltonian} function, of model parameters $\boldsymbol\theta$ and fictitious momentum parameters ${\bf{p}}\sim N({\bf 0, M})$ (with the same dimension as $\boldsymbol\theta$) as follows:
\begin{eqnarray}\label{hamiltonian}
H({\boldsymbol\theta}, {\bf p}) & = & -\log p({\boldsymbol\theta}) + \frac12 {\bf p}^{\textsf T}{\bf M}^{-1}{\bf p} 
\end{eqnarray}
where $\bf M$ is a symmetric, positive-definite \emph{mass matrix}.

Hamilton's equations, which involve differential equations of $H$, determine how ${\boldsymbol\theta}$ and ${\bf p}$ change over time. In practice, however, solving these equations exactly is too hard, so we need to approximate them by discretizing time, using some small step size $\epsilon$. For this purpose, the \emph{leapfrog} method is commonly used.

As the dimension grows, the system becomes increasingly restricted by its smallest eigen-direction, requiring smaller step sizes to maintain the stability of numerical discretization. \cite{girolami11} proposed a new method, called Riemannian Manifold HMC (RMHMC), that exploits the geometric properties of the parameter space to improve the efficiency of standard HMC. Simulating from the resulting dynamic, however, is computationally intensive since it involves solving two implicit equations, which require additional iterative numerical analysis (e.g., fixed-point iteration). 

To increase RMHMC's speed, we propose a new  integrator that is completely explicit: we propose to replace momentum with velocity in Riemannian Manifold Hamilton dynamics. As we will see, this is equivalent to using Lagrangian dynamics as opposed to Hamiltonian dynamics. By doing so, we eliminate one of the implicit steps in RMHMC. Next, we construct a time symmetric integrator to remove the remaining implicit step in RHHMC. This leads to a sampling scheme, called e-RMHMC, that involves explicit equations only. 

In what follows, we start with a brief review of RMHMC and its geometric integrator. Section 3 introduces our proposed semi-explicit integrator based on defining Hamiltonian dynamics in terms of velocity as opposed to momentum. Next, in Section 4, we eliminate the remaining implicit equation and propose a fully explicit integrator. In Section 5, we use simulated and real data to evaluate our methods' performance. Finally, in Section 6, we discuss some possible future research directions.

\section{Riemannian Manifold Hamiltonian Monte Carlo}

\begin{figure}[t]
\begin{center}
 \centerline{
\includegraphics[width=6.5in]{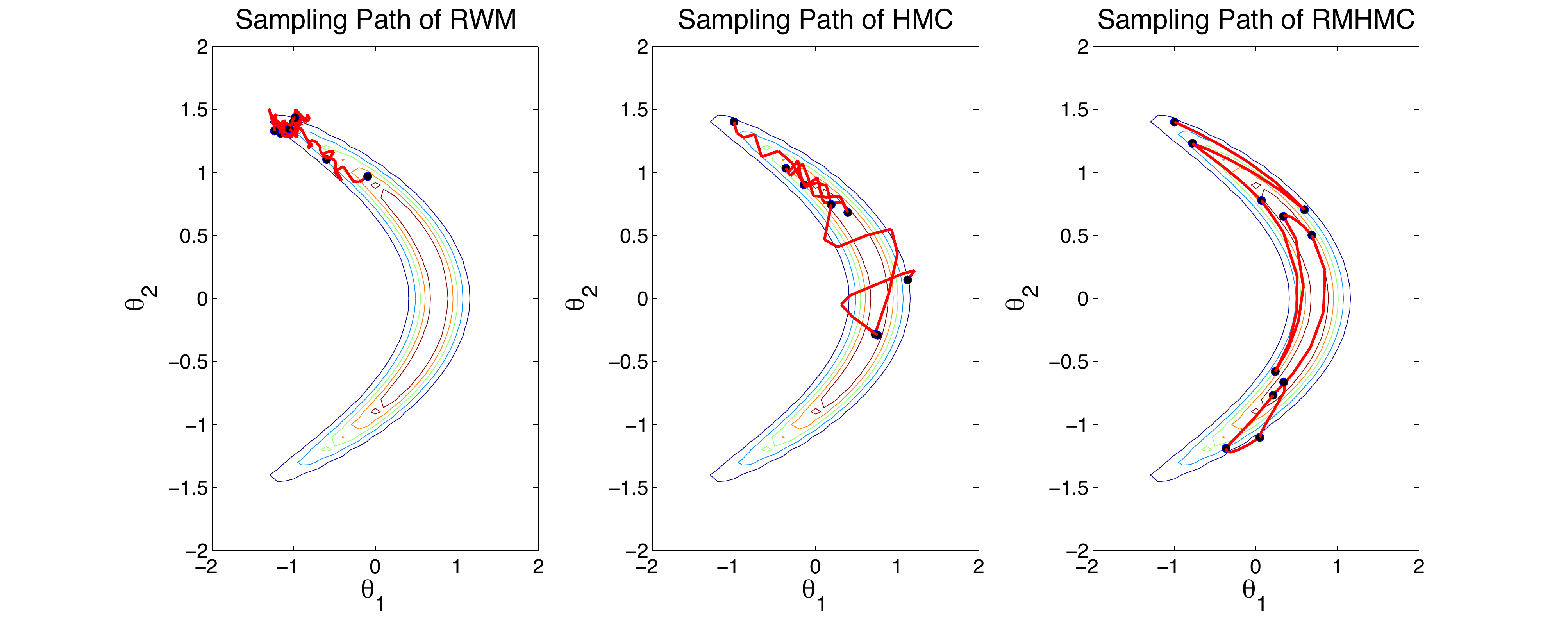}
}
\caption{The first 10 iterations in sampling from a banana shaped distribution with random walk Metropolis (RWM), Hamiltonian Monte Carlo (HMC), and Riemannian Manifold HMC (RMHMC). For all three methods, the trajectory length (i.e., step size times number of integration steps) is set to 1. Solid red lines are the sampling path, and black circles are the accepted proposals.}
\label{fig:illust}
\end{center}
 \end{figure}

As discussed above, although HMC explores the parameter space more efficiently than random walk Metropolis does, it does not fully exploits the geometric properties of parameter space defined by the density $p({\boldsymbol\theta})$. Indeed, \cite{girolami11} argue that dynamics over Euclidean space may not be appropriate to guide the exploration of parameter space. To address this issue, they propose a new method, called Riemannian Manifold HMC (RMHMC), that exploits the Riemannian geometry of the parameter space \citep{amari00} to improve standard HMC's efficiency by automatically adapting to the local structure. They do this by using a position-specific mass matrix $\bf{M} = {\bf G}({\boldsymbol\theta})$. More specifically, they set ${\bf G}({\boldsymbol\theta})$ to the Fisher information matrix. As a result, ${\bf p} = {\bf G}({\boldsymbol\theta})\dot {\boldsymbol\theta} \sim \mathcal N(0, {\bf G}({\boldsymbol\theta}))$, and Hamiltonian is defined as follows:
\begin{equation}\label{rmhamiltonp}
H({\boldsymbol\theta}, {\bf p})  = -\log p({\boldsymbol\theta}) +\frac12 \log\det {\bf G}({\boldsymbol\theta}) + 
\frac12 {\bf p}^{\textsf T}{\bf G}({\boldsymbol\theta})^{-1}{\bf p} = \phi({\boldsymbol\theta}) + \frac12 {\bf p}^{\textsf T} {\bf G}({\boldsymbol\theta})^{-1}{\bf p}
\end{equation}
where $\phi({\boldsymbol\theta}):= -\log p({\boldsymbol\theta}) +\frac12 \log\det {\bf G}({\boldsymbol\theta})$. Based on this dynamic, \cite{girolami11} propose the following HMC on Riemmanian manifold:
\begin{eqnarray}\begin{array}{lcrcr }
\displaystyle 
\dot {\boldsymbol\theta} & = & \nabla_{\bf p} H({\boldsymbol\theta}, {\bf p}) & = & {\bf G}({\boldsymbol\theta})^{-1}{\bf p} \\ [12pt]
\displaystyle 
\dot {\bf p} & = & -\nabla_{\boldsymbol\theta} H({\boldsymbol\theta}, {\bf p})& =  & -\nabla_{\boldsymbol\theta} \phi({\boldsymbol\theta}) + 
\frac12 {\boldsymbol\nu}({\boldsymbol\theta},{\bf p}) 
\end{array}\label{rmhd}\end{eqnarray}
Using the shorthand notation $\pa_i =  {\pa}/{\pa{\boldsymbol\theta}_i}$ for partial derivative, the $i$th element of the vector ${\boldsymbol\nu}({\boldsymbol\theta},{\bf p})$ is 
\begin{eqnarray*}
({\boldsymbol\nu}({\boldsymbol\theta},{\bf p}))_i = -{\bf p}^{\textsf T} \pa_i ( {\bf G}({\boldsymbol\theta})^{-1}){\bf p} = ({{\bf G}({\boldsymbol\theta})^{-1}\bf p})^{\textsf T} \pa_i {\bf G}({\boldsymbol\theta}) {{\bf G}({\boldsymbol\theta})^{-1}\bf p}
\end{eqnarray*}

The above dynamic is non-separable (it contains products of $\boldsymbol{\theta}$ and $p$), and the resulting map $({\boldsymbol\theta},{\bf p})\to ({\boldsymbol\theta}^*,{\bf p}^*)$ based on the standard leapfrog method is neither time-reversible nor symplectic. Therefore, the standard leapfrog algorithm cannot be used for the above dynamic \citep{girolami11}. Instead, we can use the St\"omer-Verlet \citep{verlet67} method as follows:
\begin{eqnarray}\label{gleapfrog}
{\bf p}^{(n+1/2)} 
 & = & {\bf p}^{(n)} - \frac{\eps}{2} \left[\nabla_{\boldsymbol\theta}\phi({\boldsymbol\theta}^{(n)})-\frac12
{\boldsymbol\nu}({\boldsymbol\theta}^{(n)},{\bf p}^{(n+1/2)})\right] \label{gleapfrog:imp1}\\
{\boldsymbol\theta}^{(n+1)} 
 & = & {\boldsymbol\theta}^{(n)} + \frac{\eps}{2} \left[{\bf G}^{-1}({\boldsymbol\theta}^{(n)}) + {\bf G}^{-1}({\boldsymbol\theta}^{(n+1)})\right]{\bf p}^{(n+1/2)}\label{gleapfrog:imp2}\\
{\bf p}^{(n+1)} 
 & = & {\bf p}^{(n+1/2)} - \frac{\eps}{2}\left[\nabla_{\boldsymbol\theta}\phi({\boldsymbol\theta}^{(n+1)})-\frac12
{\boldsymbol\nu}({\boldsymbol\theta}^{(n+1)},{\bf p}^{(n+1/2)})\right]
\end{eqnarray}
This is also known as generalized leapfrog \citep{leimkuhler04}. The above series of transformations are (i) deterministic (ii) reversible and (iii) volume-preserving. Therefore, the effective proposal distribution is a delta function $\delta(({\boldsymbol\theta}^{(1)}, {\bf p}^{(1)}), ({\boldsymbol\theta}^{(L)}, {\bf p}^{(L)})$ and the acceptance probability is as follows:
\begin{equation}
\frac{\exp(-H({\boldsymbol\theta}^{(L)}, {\bf p}^{(L)}))}{\exp(-H({\boldsymbol\theta}^{(1)}, {\bf p}^{(1)}))} 
\times 
\frac
{\delta(({\boldsymbol\theta}^{(L)}, {\bf p}^{(L)}), ({\boldsymbol\theta}^{(1)}, {\bf p}^{(1)}))}
{\delta(({\boldsymbol\theta}^{(1)}, {\bf p}^{(1)}), ({\boldsymbol\theta}^{(L)}, {\bf p}^{(L)}))} = 
\exp(H({\boldsymbol\theta}^{(1)}, {\bf p}^{(1)})-H({\boldsymbol\theta}^{(L)}, {\bf p}^{(L)}))
\end{equation}
Here, $({\boldsymbol\theta}^{(1)}, {\bf p}^{(1)})$ is the current state, and $({\boldsymbol\theta}^{(L)}, {\bf p}^{(L)})$ is the proposal after $L$ leapfrog steps. 

As an illustrative example, Figure \ref{fig:illust} shows the sampling paths of random walk Metropolis (RWM), HMC, and RMHMC for an artificially created banana-shaped distribution \citep[See ][discussion by Luke Bornn and Julien Cornebise]{girolami11}. For this example, we fixed the trajectory and chose the step sizes such that the acceptance probability for all three methods remains around 0.7. RWMH moves slowly and spends most of iterations at the distribution's low-density tail, and HMC explores the parameter space in a tortuous way, while RMHMC moves directly to the high-density region and explores the distribution more efficiently.

One major drawback of this geometric integrator, which is both time-reversible and volume-preserving, is that it involves two implicit functions: Equations \eqref{gleapfrog:imp1} and \eqref{gleapfrog:imp2}. These functions require extra numerical analysis (e.g. fixed-point iteration), which results in higher computational cost and simulation error. To address this problem, we propose an alternative approach that uses velocity instead of momentum.

\section{Moving from Momentum to Velocity}\label{useV}
In the Hamiltonian dynamic \eqref{rmhd}, the product of ${\bf G}({\boldsymbol\theta})^{-1}$ and ${\bf p}$ is in fact velocity, ${\bf v}={\bf G}({\boldsymbol\theta})^{-1}{\bf p}$. This motivates us to define the dynamic in terms of ${\bf v}$ instead of ${\bf p}$. The transformation ${\bf p}\mapsto {\bf v}$ changes the Hamiltonian dynamics \eqref{rmhd} to the following form (derivation in Appendix \ref{derRMLD}):
\begin{eqnarray}\begin{array}{lcl }
\displaystyle 
\dot {\boldsymbol\theta} & = & {\bf v} \\[12pt]
\displaystyle 
\dot {\bf v} & = & - {\boldsymbol\eta}({\boldsymbol\theta},{\bf v}) - {\bf G}({\boldsymbol\theta})^{-1} \nabla_{\boldsymbol\theta} \phi({\boldsymbol\theta})
\end{array}\label{rmld}\end{eqnarray}
where ${\boldsymbol\eta}({\boldsymbol\theta},{\bf v})$ is a vector whose $k$th element is $\sum\nolimits_{i,j}\Gamma^k_{ij}({\boldsymbol\theta}){\bf v}^i {\bf v}^j$. Here, $\Gamma^k_{ij}({\boldsymbol\theta}) := \frac12 \sum_l {\bf g}^{kl} (\pa_i {\bf g}_{lj} + \pa_j {\bf g}_{il} - \pa_l {\bf g}_{ij}) $ is Christoffel symbol whose $(i,j)$th element is ${\bf G}({\boldsymbol\theta})= ({\bf g}_{ij})$. Further, ${\bf G}({\boldsymbol\theta})^{-1}=({\bf g}^{ij})$.

This transformation moves the Hamiltonian dynamic's complexity from its first equation for ${\boldsymbol\theta}$ to its second equation. In this way, we resolve one implicit function in the generalized leapfrog method and develop a semi-explicit integrator as follows:
\begin{eqnarray}\label{GLv0}
{\bf v}^{(n+1/2)}  & = & {\bf v}^{(n)} - \frac{\eps}{2}[({\bf v}^{(n+1/2)})^{\textsf T} \Gamma({\boldsymbol\theta}^{(n)}) {\bf v}^{(n+1/2)} + {\bf G}({\boldsymbol\theta}^{(n)})^{-1} \nabla_{\boldsymbol\theta}\phi({\boldsymbol\theta}^{(n)})] \label{GLv:1}\\
{\boldsymbol\theta}^{(n+1)} & = & {\boldsymbol\theta}^{(n)} + \eps {\bf v}^{(n+1/2)}  \label{GLv:2}\\
{\bf v}^{(n+1)} & = & {\bf v}^{(n+1/2)} - \frac{\eps}{2}[({\bf v}^{(n+1/2)})^{\textsf T} \Gamma({\boldsymbol\theta}^{(n+1)}) {\bf v}^{(n+1/2)} + {\bf G}({\boldsymbol\theta}^{(n+1)})^{-1} \nabla_{\boldsymbol\theta}\phi({\boldsymbol\theta}^{(n+1)})] \label{GLv:3}
\end{eqnarray}
Note that updating ${\bf v}$ remains implicit (more details are available in Appendix \ref{semi-RMLMC}).

In general, the new dynamic \eqref{rmld} cannot be recognized as a Hamiltonian dynamic of $({\boldsymbol\theta}, {\bf v})$. Nevertheless, it remains a valid proposal-generating mechanism, which preserves the original Hamiltonian $H({\boldsymbol\theta}, {\bf p}={\bf G}({\boldsymbol\theta}){\bf v})$ (proof in Appendix \ref{derRMLD}); thus, the acceptance probability is only determined by a discretization error from the numerical integration, as before.

Because ${\bf p}\sim \mathcal N({\bf 0},{\bf G}({\boldsymbol\theta}))$, the distribution of ${\bf v}={\bf G}({\boldsymbol\theta})^{-1}{\bf p}$ is $ \mathcal N({\bf 0},{\bf G}({\boldsymbol\theta})^{-1})$. Therefore, we have
\begin{eqnarray*}
p({\bf v}) = \frac{1}{(\sqrt{2\pi})^{D}\sqrt{\det({\bf G}({\boldsymbol\theta})^{-1})}} \exp\left\{-\frac12 {\bf v}^{\textsf T} {\bf G}({\boldsymbol\theta}) {\bf v} \right\} \propto (\det {\bf G}({\boldsymbol\theta}))^{1/2} \exp\left\{-\frac12 {\bf v}^{\textsf T} {\bf G}({\boldsymbol\theta}) {\bf v} \right\}
\end{eqnarray*}
We define the \emph{energy} function ${\bf E}({\boldsymbol\theta}, {\bf v})$ as the sum of the potential energy, $U({\boldsymbol\theta})$ and $K({\boldsymbol\theta}, {\bf v})$, where $K({\boldsymbol\theta}, {\bf v}) = -\log(p({\bf v}))$. 

Analogous to RMHMC, thebacceptance probability is calculated based on ${\bf E}({\boldsymbol\theta}, {\bf v})$, which is the negative log of the joint density of parameter ${\boldsymbol\theta}$ and the new auxiliary variable ${\bf v}$. Therefore, we can apply the generalized leapfrog scheme to this new dynamic \eqref{rmld} and derive a semi-implicit method. (See more details in Appendix \ref{semi-RMLMC}.) Although the resulting integrator is not symplectic, we nonetheless have detailed balance with volume correction. Algorithm \ref{Alg:RMLMC} shows the corresponding steps for implementing this method.

\begin{algorithm}[t]
\caption{Semi-explicit Riemannian Manifold Lagrangian Monte Carlo (RMLMC)}
\label{Alg:RMLMC}
\begin{algorithmic}
\STATE Initialize ${\boldsymbol\theta}^{(1)} = \textrm{current}\; {\boldsymbol\theta}$
\STATE Sample new velocity ${\bf v}^{(1)}\sim \mathcal N(0,{\bf G}^{-1}({\boldsymbol\theta}^{(1)}))$
\STATE Calculate current ${\bf E}({\boldsymbol\theta}^{(1)}, {\bf v}^{(1)})$ according to equation \eqref{energyv}
\FOR{$n=1$ to $L$ (leapfrog steps)}
\STATE \% Update the velocity with fixed point iterations
\STATE ${\bf\hat v}^{(0)} = {\bf v}^{(n)}$
\FOR{$i=1$ to NumOfFixedPointSteps}
\STATE ${\bf\hat v}^{(i)} = {\bf v}^{(n)} - \frac{\eps}{2}{\bf G}({\boldsymbol\theta}^{(n)})^{-1}[({\bf\hat v}^{(i-1)})^{\textsf T}\tilde\Gamma({\boldsymbol\theta}^{(n)}) {\bf\hat v}^{(i-1)} + \nabla_{\boldsymbol\theta}\phi({\boldsymbol\theta}^{(n)})]$
\ENDFOR
\STATE ${\bf v}^{(n+1/2)} = {\bf\hat v}^{(last\; i)}$
\STATE \% {\it Update the position only with simple one step}
\STATE {\boldmath ${\boldsymbol\theta}^{(n+1)} = {\boldsymbol\theta}^{(n)} + \eps {\bf v}^{(n+1/2)}$}
\STATE $\Delta\log\det_n = \log \det ({\bf I}-\eps ({\bf v}^{(n+1/2)})^{\textsf T} \Gamma({\boldsymbol\theta}^{(n+1)})) - \log \det ({\bf I}+\eps ({\bf v}^{(n+1/2)})^{\textsf T} \Gamma({\boldsymbol\theta}^{(n)}))$
\STATE Update the velocity exactly
\STATE ${\bf v}^{(n+1)} = {\bf v}^{(n+1/2)} - \frac{\eps}{2}{\bf G}({\boldsymbol\theta}^{(n+1)})^{-1}[({\bf v}^{(n+1/2)})^{\textsf T}\tilde\Gamma({\boldsymbol\theta}^{(n+1)}) {\bf v}^{(n+1/2)} + \nabla_{\boldsymbol\theta}\phi({\boldsymbol\theta}^{(n+1)})]$
\ENDFOR
\STATE Calculate proposed ${\bf E}({\boldsymbol\theta}^{(L+1)}, {\bf v}^{(L+1)})$ according to equation \eqref{energyv}
\STATE logRatio = $-\textrm{ProposedH}+\textrm{CurrentH} + \sum_{n=1}^{N}\Delta\log\det_n$
\STATE Accept or reject according to Metropolis ratio
\end{algorithmic}
\end{algorithm}

In Appendix \ref{connection2L}, we show that the new dynamic \eqref{rmld} is essentially a \emph{Lagrangian} dynamic. Therefore, we refer to the derived algorithm (Algorithm \ref{Alg:RMLMC}) as Riemannian Manifold Lagrangian Monte Carlo (RMLMC), which explores the parameter space along the path on a Riemannian manifold that minimizes the total Lagrangian. We can use this new proposal-generating mechanism, RMLMC, which is based on an Euler-Lagrange system \eqref{rmld} of $({\boldsymbol\theta}, {\bf v})$ instead of the original Hamiltonian system \eqref{rmhd} defined in terms of $({\boldsymbol\theta}, {\bf p})$. The two methods use equivalent dynamics but differ numerically in the following way: RMHMC augments parameter space with momentum, while RMLMC augments parameter space with velocity. Later, we will show that switching to velocity leads to substantial improvement in computational efficiency.

\section{Explicit Riemannian Manifold Lagrangian Monte Carlo}
We now propose a fully explicit integrator for Lagrangian dynamics \eqref{rmld} as follows: 
\begin{eqnarray}\label{TSvsolved}
{\bf v}^{(n+1/2)} & = & [{\bf I} + \frac{\eps}{2} {\bf\Omega}({\boldsymbol\theta}^{(n)},{\bf v}^{(n)})]^{-1}
[{\bf v}^{(n)} - \frac{\eps}{2} {\bf G}({\boldsymbol\theta}^{(n)})^{-1}
\nabla_{\boldsymbol\theta}\phi({\boldsymbol\theta}^{(n)})]  \label{TSvsolved:1}\\
{\boldsymbol\theta}^{(n+1)} &=& {\boldsymbol\theta}^{(n)} + \epsilon{\bf v}^{(n+1/2)} \\
{\bf v}^{(n+1)} & = & [{\bf I} + \frac{\eps}{2} {\bf\Omega}({\boldsymbol\theta}^{(n + 1)},{\bf v}^{(n + \frac 1 2)})]^{-1}[{\bf v}^{(n+1/2)} - \frac{\eps}{2} {\bf G}({\boldsymbol\theta}^{(n+1)})^{-1}\nabla_{\boldsymbol\theta}\phi({\boldsymbol\theta}^{(n+1)})] \label{TSvsolved:2}
\end{eqnarray}
where ${\bf\Omega}({\boldsymbol\theta}^{(n)},{\bf v}^{(n)})$ is a matrix whose $(i,j)$th element is $\sum_k {\bf v}^{(n)}_k\Gamma_{kj}^i({\boldsymbol\theta}^{(n)})$. This integrator is (i) reversible and (ii) energy-preserving up to order ${\cal O}(\eps)$, where $\eps$ is the stepsize. The resulting map, however, is not volume-preserving and as such the effective proposal distribution will be the product of a delta function and the determinant of the transformation,
\begin{equation*}
\frac{\exp(-{\bf E}({\boldsymbol\theta}^{(L)}, {\bf v}^{(L)}))}{\exp(-{\bf E}({\boldsymbol\theta}^{(1}, {\bf v}^{(1)}))} 
\times 
\frac
{\delta(({\boldsymbol\theta}^{(L)}, {\bf v}^{(L)}), ({\boldsymbol\theta}^{(1)}, {\bf v}^{(1)}))}
{\delta(({\boldsymbol\theta}^{(1)}, {\bf v}^{(1)}), ({\boldsymbol\theta}^{(L)}, {\bf v}^{(L)}))} \times \det {\bf J}  
\end{equation*}
which simplifies to
\begin{equation}\label{acptsv}
\exp({\bf E}({\boldsymbol\theta}^{(1)}, {\bf v}^{(1)})-{\bf E}({\boldsymbol\theta}^{(L)}, {\bf v}^{(L)})) \times \det {\bf J}
\end{equation}
with ${\bf J}$ the Jacobian matrix of $({\boldsymbol\theta}^{(1)}, {\bf v}^{(1)})\to ({\boldsymbol\theta}^{(L)}, {\bf v}^{(L)})$. Detailed derivations and proofs are given in Appendix \ref{xplct-RMLMC}.

\begin{algorithm}[t]
\caption{Explicit Riemannian Manifold Lagrangian Monte Carlo (e-RMLMC)}
\label{Alg:e-RMLMC}
\begin{algorithmic}
\STATE Initialize ${\boldsymbol\theta}^{(1)} = \textrm{current}\; {\boldsymbol\theta}$
\STATE Sample new velocity ${\bf v}^{(1)}\sim \mathcal N(0,{\bf G}({\boldsymbol\theta}^{(1)})^{-1})$
\STATE Calculate current ${\bf E}({\boldsymbol\theta}^{(1)},{\bf v}^{(1)})$ according to equation \eqref{energyv}
\STATE $\Delta \log\det = 0$
\FOR{$n=1$ to $L$}
\STATE $\Delta \log\det = \Delta \log\det - \det({\bf G}({\boldsymbol\theta}^{(n)})+\eps/2 {\bf\tilde\Omega}({\boldsymbol\theta}^{(n)},{\bf v}^{(n)}))$
\STATE Update the velocity {\bf explicitly} with a half step:
\STATE ${\bf v}^{(n+1/2)}\! =\! [{\bf G}({\boldsymbol\theta}^{(n)})\!+\!\frac{\eps}{2} {\bf\tilde\Omega}({\boldsymbol\theta}^{(n)},{\bf v}^{(n)})]^{-1}[{\bf G}({\boldsymbol\theta}^{(n)}){\bf v}^{(n)}\!-\!\frac{\eps}{2} \nabla_{\boldsymbol\theta}\phi({\boldsymbol\theta}^{(n)})]$
\STATE $\Delta \log\det = \Delta \log\det + \det({\bf G}({\boldsymbol\theta}^{(n)})-\eps/2 {\bf\tilde\Omega}({\boldsymbol\theta}^{(n)},{\bf v}^{(n+1/2)}))$
\STATE Update the position with a full step:
\STATE {\boldmath ${\boldsymbol\theta}^{(n+1)} = {\boldsymbol\theta}^{(n)} + \eps {\bf v}^{(n+\frac{1}{2})}$}
\STATE $\Delta \log\det = \Delta \log\det - \det({\bf G}({\boldsymbol\theta}^{(n+1)})+\eps/2 {\bf\tilde\Omega}({\boldsymbol\theta}^{(n+1)},{\bf v}^{(n+1/2)}))$
\STATE Update the velocity {\bf explicitly} with a half step:
\STATE ${\bf v}^{(n+1)}\! =\! [{\bf G}({\boldsymbol\theta}^{(n+1)})\!+\!\frac{\eps}{2} {\bf\tilde\Omega}({\boldsymbol\theta}^{(n+1)},{\bf v}^{(n+1/2)})]^{-1}[{\bf G}({\boldsymbol\theta}^{(n+1)}){\bf v}^{(n+1/2)}\!-\!\frac{\eps}{2} \nabla_{\boldsymbol\theta}\phi({\boldsymbol\theta}^{(n+1)})]$
\STATE $\Delta \log\det = \Delta \log\det + \det({\bf G}({\boldsymbol\theta}^{(n+1)})-\eps/2 {\bf\tilde\Omega}({\boldsymbol\theta}^{(n+1)},{\bf v}^{(n+1)}))$
\ENDFOR
\STATE Calculate proposed ${\bf E}({\boldsymbol\theta}^{(L+1)},{\bf v}^{(L+1)})$ according to equation \eqref{energyv}
\STATE logRatio = $-\textrm{Proposed{\bf E}}+\textrm{Current{\bf E}} + \Delta \log\det$
\STATE Accept or reject the proposed state according to \eqref{acptsv}
\end{algorithmic}
\end{algorithm}

We refer to this approach as explicit Riemannian Manifold Lagrangian Monte Carlo (e-RMLMC). Algorithm \ref{Alg:e-RMLMC} shows the corresponding steps for this method. In this algorithm, we use ${\bf\tilde\Omega}({\boldsymbol\theta},{\bf v})$ to denote ${\bf G}({\boldsymbol\theta}){\bf\Omega}({\boldsymbol\theta},{\bf v})$ whose $(k,j)$th element is equal to $\sum_{i} {\bf v}^{i}\tilde\Gamma_{ij}^{k}({\boldsymbol\theta})$.

Our proposed e-RMLMC does not involve implicit functions for updating $({\boldsymbol\theta},{\bf v})$ in RMHMC. Because we remove multiple fixed-point iteration steps, we reduce the computation time by $\mathcal O(D^{2})$ where $D$ is the dimension of the parameters. Additionally, using this explicit updating, we resolve the convergence issue faced by fixed-point iterations. The connection terms $\tilde \Gamma({\boldsymbol\theta})$ in ${\bf\tilde\Omega}$ do not add substantial computational cost since they are obtained from permuting three dimensions of the array $\pa {\bf G}({\boldsymbol\theta})$, which is computed in RMHMC. However, besides ${\bf G}({\boldsymbol\theta})^{-1}$, which is required for $\nabla_{\boldsymbol\theta}\phi({\boldsymbol\theta})$, e-RMLMC has two extra matrix inversions to update ${\bf v}$, whose complexity in general is $\mathcal O(D^{3})$. Therefore, as dimension grows, the efficiency gained by removing multiple fixed-point iterations may be overwhelmed by this additional overhead. This is evident from our experimental results presented in Section \ref{experiments}. Faster matrix inversion algorithms could be used to alleviate this issue.

\section{Experimental Results} \label{experiments}
In this section, we use simulated and real data to evaluate our methods, RMLMC and e-RMLMC, compared to RMHMC. Following \cite{girolami11}, we use a time-normalized effective sample size (ESS) to compare these methods. For $B$ posterior samples we calculate ESS = $B[1 + 2\Sigma_{k}\gamma(k)]^{-1}$ for each parameter and choose the minimum as the measure of sampling efficiency, where $\Sigma_{k}\gamma(k)$ is the sum of the $K$ monotone sample autocorrelations estimated by the initial monotone sequence estimator \citep{geyer92}. All computer programs and data sets discussed in this paper are available online at \url{http://www.ics.uci.edu/~babaks/Site/Codes.html}.

\begin{table}[!htbp]
\begin{center}
\begin{tabular}{l|l|cccc}
  \hline
Data & method & AP & s & ESS & min(ESS)/s \\ 
  \hline
 & RMHMC & 0.74 & 2.40E-02 & (8561,9595,10262) & 23.77 \\ 
  Australian & RMLMC & 0.75 & 1.90E-02 & (8038,10488,11468) & 28.24 \\ 
  D=14,N=690 & e-RMLMC & 0.75 & 1.55E-02 & (9636,10443,11268) & 41.34 \\ 
   \hline
 & RMHMC & 0.77 & 5.63E-02 & (15000,15000,15000) & 17.76 \\ 
  German & RMLMC & 0.73 & 4.32E-02 & (15000,15000,15000) & 23.17 \\ 
  D=24,N=1000 & e-RMLMC & 0.70 & 3.59E-02 & (13762,15000,15000) & 25.57 \\ 
   \hline
 & RMHMC & 0.77 & 1.65E-02 & (7050,8369,8905) & 28.57 \\ 
  Heart & RMLMC & 0.78 & 1.09E-02 & (10847,11704,12405) & 66.25 \\ 
  D=13,N=270 & e-RMLMC & 0.76 & 1.01E-02 & (10347,10724,11773) & 68.18 \\ 
   \hline
 & RMHMC & 0.81 & 1.25E-02 & (4325,4622,4980) & 23.10 \\ 
  Pima & RMLMC & 0.82 & 7.27E-03 & (4713,5448,5576) & 43.20 \\ 
  D=7,N=532 & e-RMLMC & 0.82 & 7.04E-03 & (4839,5193,5539) & 45.85 \\ 
   \hline
 & RMHMC & 0.78 & 8.39E-03 & (15000,15000,15000) & 119.20 \\ 
  Ripley & RMLMC & 0.76 & 5.07E-03 & (13498,15000,15000) & 177.43 \\ 
  D=2,N=250 & e-RMLMC & 0.79 & 4.77E-03 & (12611,15000,15000) & 176.37 \\ 
   \hline
\end{tabular}
\caption{Comparing alternative methods using five binary classification problems discussed in \cite{girolami11}. For each dataset, the number of predictors, $D$, and the number of observations, $N$, are specified. For each method, we provide the acceptance probability (AP), the CPU time (s) for each iteration, and the time-normalized ESS.}
\label{realLR}
\end{center}
\end{table}

\subsection{Logistic Regression Models}
We start by evaluating our methods based on five binary classification problems used in \cite{girolami11}. These are Australian Credit data, German Credit data, Heart data, Pima Indian data, and Ripley data. For each problem, we use a logistic regression model and run 20000 MCMC iterations. Results (after discarding the initial 5000 iterations) are summarized in Table \ref{realLR}, and show that in general our methods improve the sampling efficiency measured in terms of ESS per second compared to RMHMC.

\subsection{Simulated Logistic Regression}
Next, we construct some synthetic datasets with variable numbers of observations and increasing dimensionality for logistic regression. The simulated results are summarized in Table \ref{simLR}. In general, our methods RMLMC and e-RMLMC improve RMHMC in minimal ESS per second, but as expected, such an advantage gradually diminishes as the dimension increases.

\begin{table}[!htbp]
\begin{center}
	\begin{tabular} {l|l|cccc}
		\hline
			Data & Method & AP & s & ESS & min(ESS)/s \\
		\hline
			\multirow{3}{*}{N=200,D=10} & RMHMC & 0.84 & 1.85e+02 & (4837, 4902, 4968) & 26.20 \\ 
			 & RMLMC & 0.86 & 4.44e+01 & (5000, 5000, 5000) & 112.60 \\ 
			 & e-RMLMC & 0.83 & 7.42e+01 & (3792, 4310, 4671) & 51.11 \\
		\hline
			\multirow{3}{*}{N=400,D=20} & RMHMC & 0.82 & 9.56e+02 & (4727, 4893, 5000) & 4.95 \\ 
			 & RMLMC & 0.80 & 1.28e+02 & (4680, 4819, 4857) & 36.44 \\ 
			 & e-RMLMC & 0.81 & 2.59e+02 & (2964, 3543, 3968) & 11.46 \\
		\hline
			\multirow{3}{*}{N=800,D=40} & RMHMC & 0.82 & 3.15e+03 & (4691, 4983, 5000) & 1.49 \\ 
			 & RMLMC & 0.82 & 6.88e+02 & (4749, 4836, 4960) & 6.91 \\
			 & e-RMLMC & 0.81 & 1.09e+03 & (2902, 3636, 4127) & 2.65 \\
		\hline
			\multirow{3}{*}{N=1600,D=80} & RMHMC & 0.81 & 9.87e+03 & (3712, 4515, 4950) & 0.38 \\ 
			 & e-RMLMC & 0.83 & 4.64e+03 & (4002, 4672, 4919) & 0.86 \\
			 & e-RMLMC & 0.80 & 1.19e+04 & (2565, 3415, 4081) & 0.22 \\
		\hline 
			\multirow{3}{*}{N=3200,D=160} & RMLMC & 0.79 & 1.63e+05 & (3160, 3959, 4464) & 0.02 \\ 
			 & RMLMC & 0.83 & 1.44e+05 & (3458, 4221, 4676) & 0.02 \\
			 & e-RMLMC & 0.80 & 1.20e+05 & (2708, 3548, 4156) & 0.02 \\
		\hline
	\end{tabular}
	\caption{Time, ESS and time-normalized ESS for logistic regression with simulated datasets. Results are calculated on a 5,000 sample chain with a 5,000 sample burn-in session.}
	\label{simLR}
\end{center}
\end{table}

\subsection{Simulating a banana-shaped distribution}
The banana-shaped distribution, which we used above for illustration, can be constructed as the posterior distribution of $\theta=(\theta_{1},\theta_{2})|y$ based on the following model:
\begin{eqnarray*}
y|\theta & \sim & N(\theta_{1}+\theta_{2}^{2},\sigma^{2}_{y})\\
\theta & \sim & N(0,\sigma^{2}_{\theta})
\end{eqnarray*}
The data $\{y_{i}\}_{i=1}^{100}$ are generated with $\theta_{1}+\theta_{2}^{2}=1,\sigma_{y}=2$. We set $\sigma_{\theta}=1$.

We want to investigate how the three algorithms, RMHMC, RMLMC, e-RMLMC, explore the parameter space. Fig.\ref{fig:comp3} shows the first 10 iterations for each algorithm using fixed trajectory length of 1.45.
\begin{figure}[!t]
 \centerline{
\includegraphics[width=6.5in]{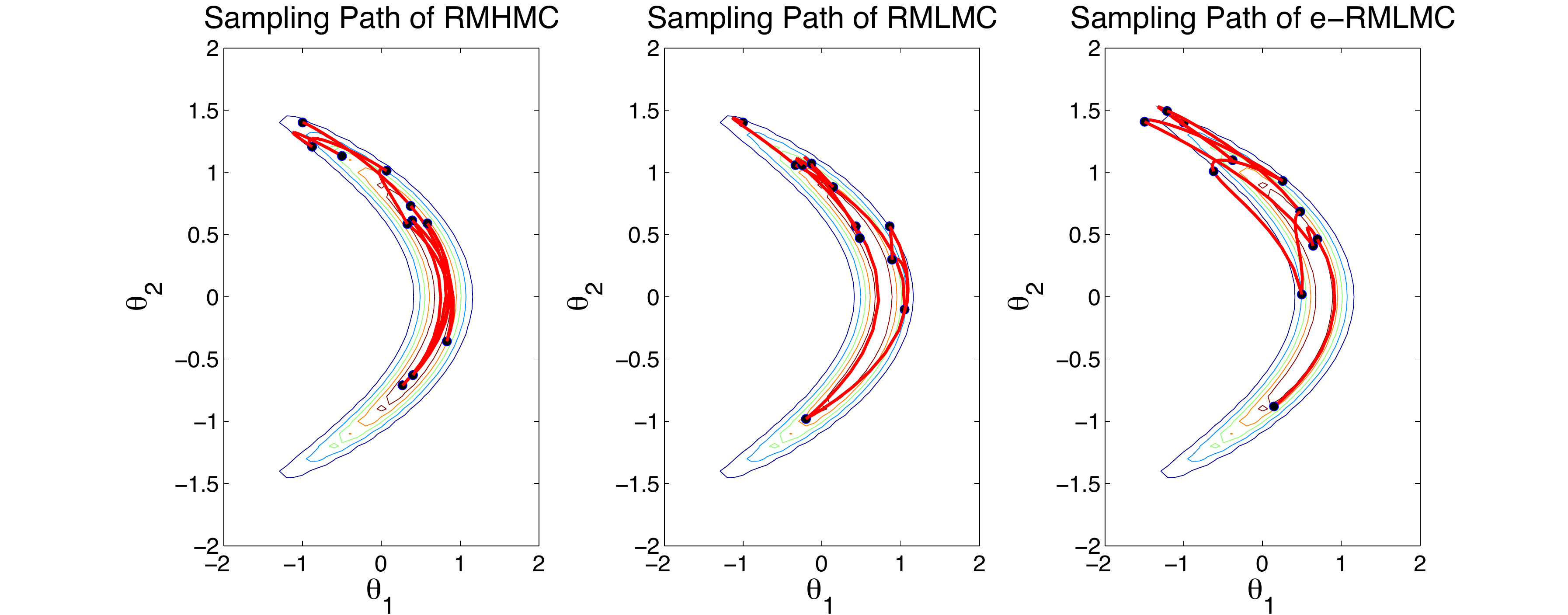}
}
\caption{The first 10 iterations in sampling from the banana-shaped distribution with Riemannian Manifold HMC (RMHMC), Riemannian Manifold Lagrange Monte Carlo (RMLMC) and explicit RMLMC (e-RMLMC). For all three methods, the trajectory length (i.e., step size times number of integration steps) is set to 1.45. Solid red lines show the sampling path, and each point represents an accepted proposal.}
\label{fig:comp3}
 \end{figure}
We can see that RMLMC and e-RMLMC explore the parameter space according to its curvature and mix quickly, similar to RMHMC. 

Table \ref{toy} compares the performances of these algorithms based on 5000 MCMC iterations (after burning the initial 1000 iteration). For this example, RMLMC has the highest ESS/s. As discussed above, the additional overhead of matrix inversion for e-RMLMC occasionally overwhelms its gain in computational efficiency. 
\begin{table}[!htbp]
\begin{center}
\begin{tabular}{ccccc}
  \hline
Method &  AP & s & ESS & min(ESS)/s \\ 
  \hline
  RMHMC & 0.73 & 8.21e-03 & (729,1117,1506) & 17.76 \\ 
  RMLMC & 0.79 & 5.36e-03 & (857,1317,1777) & 31.99 \\ 
  e-RMLMC & 0.78 & 5.61e-03 & (585,1085,1585) & 20.85 \\ 
   \hline
\end{tabular}
\caption{Comparing alternative methods using a banana-shaped distribution. For each method, we provide the acceptance probability (AP), the CPU time (s) for each iteration, and the time-normalized ESS.}
\label{toy}
\end{center}
\end{table}

\subsection{Finite Mixture of Gaussians}
Finally we consider finite mixtures of univariate Gaussian components of the form 
\begin{equation}
p(x_i|\boldsymbol{\theta}) = \sum_{k=1}^K\pi_k\mathcal{N}(x_i|\mu_k,\sigma^2_k) 
\label{eq:likelihood}
\end{equation}
where $\boldsymbol{\theta}$ is the vector of size $D=3K$ of all the parameters $\pi_k$, $\mu_k$ and $\sigma^2_k$ and $\mathcal{N}(\cdot|\mu,\sigma^2)$ is a Gaussian density with mean $\mu$ and variance $\sigma^2$. A common choice of prior takes the form 
\begin{equation}
p(\boldsymbol{\theta}) = \mathcal{D}(\pi_1,\dots,\pi_K | \lambda)\prod_{k=1}^K \mathcal{N}(\mu_k|m,\beta^{-1}\sigma_k^2)\mathcal{IG}(\sigma^2_k| b,c)
\label{eq:prior}
\end{equation}
where $ \mathcal{D}(\cdot|\lambda)$ is the symmetric Dirichlet distribution with parameter $\lambda$ and $\mathcal{IG}(\cdot|b,c)$ is the inverse Gamma distribution with shape parameter $b$ and scale parameter $c$. 

Although the posterior distribution associated with this model is formally explicit, it is computationally intractable, since it can be expressed as a sum of $K^N$ terms corresponding to all possible allocations of observations $x_i$ to mixture components \cite[chap. 9]{marin05}. We want to use this model to test the efficiency of posterior sampling $\boldsymbol{\theta}$ using the three methods. A more extensive comparison of Riemannian Manifold MCMC and HMC, Gibbs sampling and standard Metropolis-Hastings for finite Gaussian mixture models can be found at \cite{Stathopoulos11}.
Due to the non-analytic nature of the expected Fisher Information, $\boldsymbol{I}(\boldsymbol{\theta})$, we use the empirical Fisher information as metric tensor, defined in \cite[chap. 2]{mclachlan00}:
\begin{equation*}
\boldsymbol{G}(\boldsymbol{\theta}) = \boldsymbol{S}^T\boldsymbol{S}  - \frac{1}{N}{\boldsymbol{s}}{\boldsymbol{s}}^T \label{eq:eFI}
\end{equation*}
where $N\times D$ score matrix $\boldsymbol{S}$ has elements $S_{i,d}=\frac{\partial \log p(x_i|\boldsymbol{\theta})}{\partial \theta_d}$ and ${\boldsymbol{s}}=\sum_{i=1}^N\boldsymbol{S}^T_{i,\cdot}$.

Depending on allocations of $\pi$, we show several classical mixtures showing in the following Table \ref{tab:dens} and Figure \ref{fig:dens}. Their sampling efficiency is compared in Table \ref{MOG}. As before, our two algorithms outperform RMHMC. 

\begin{table}[t]
\begin{center}
\begin{small} 
\begin{tabular}{@{}ccc@{}} 
\toprule 
Dataset & Density function & Num. of \\
name &  & parameters \\ 
\midrule 
Kurtotic & $ \frac{2}{3}\mathcal{N}(x | 0,1)+\frac{1}{3}\mathcal{N}\left(x|0,\left(\frac{1}{10}\right)^2 \right)$ & 6\\
Bimodal & $\frac{1}{2}\mathcal{N}\left(x | -1,\left(\frac{2}{3}\right)^2\right)+\frac{1}{2}\mathcal{N}\left(x|1,\left(\frac{2}{3}\right)^2 \right) $ & 6 \\
Skewed & $ \frac{3}{4}\mathcal{N}\left(x | 0,1 \right)+\frac{1}{4}\mathcal{N}\left(x|\frac{3}{2},\left(\frac{1}{3}\right)^2 \right)$ & 6\\
Trimodal & $\frac{9}{20}\mathcal{N}\left(x | -\frac{6}{5},\left(\frac{3}{5}\right)^2 \right)+\frac{9}{20}\mathcal{N}\left(x|\frac{6}{5},\left(\frac{3}{5}\right)^2 \right)+\frac{1}{10}\mathcal{N}\left(x|0,\left(\frac{1}{4}\right)^2 \right) $ & 9\\
Claw & $\frac{1}{2}\mathcal{N}\left(x |0,1\right)+\sum_{i=0}^4\frac{1}{10}\mathcal{N}\left(x|\frac{i}{2}-1,\left(\frac{1}{10}\right)^2 \right) $ &18 \\
\toprule 
\end{tabular}
\end{small} 
\caption{Densities used for the generation of synthetic Mixture of Gaussian data sets. \label{tab:dens}}{} 
\end{center}
\end{table}

\begin{figure}[t]
\begin{center}
\includegraphics[width=0.19\textwidth]{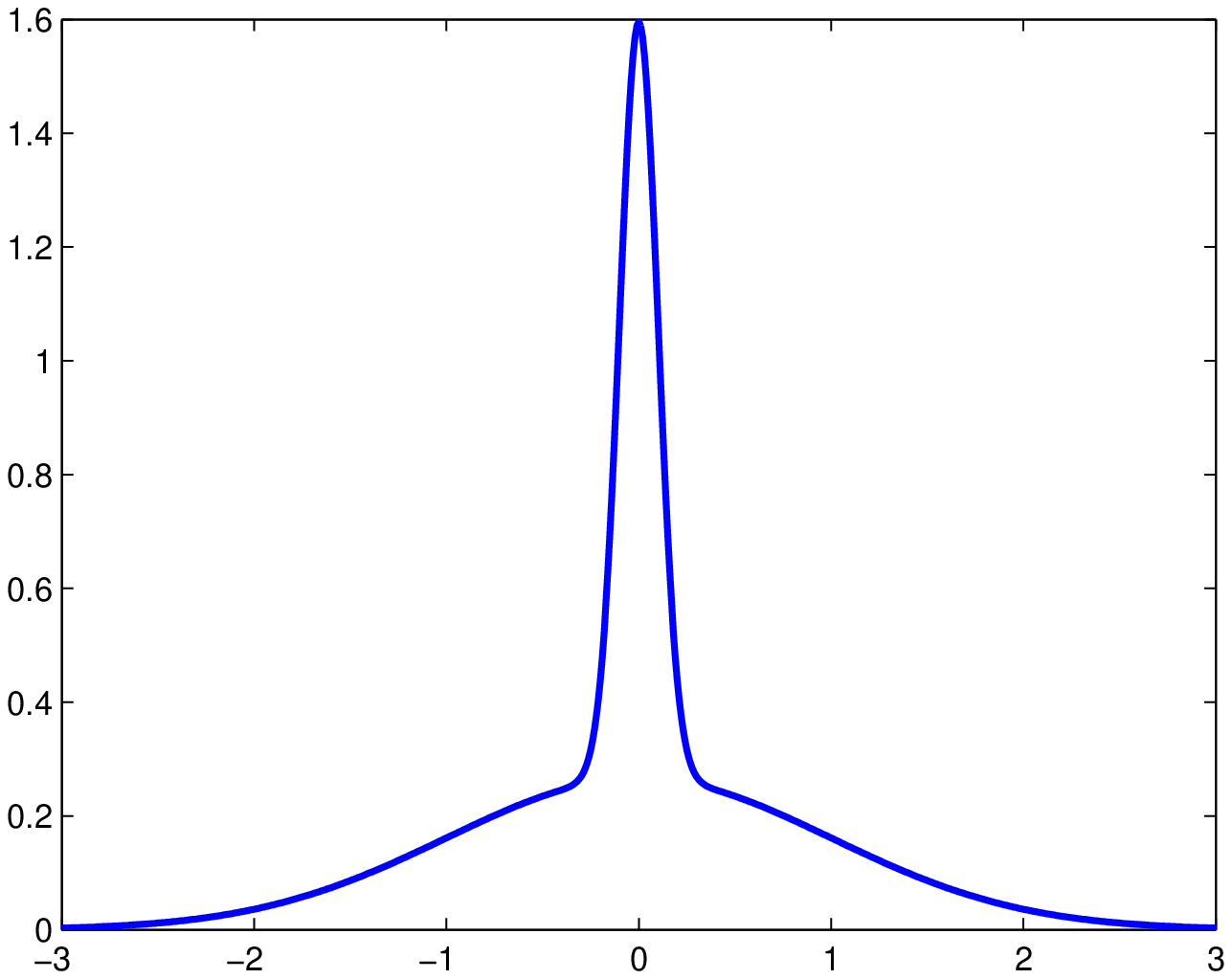}\includegraphics[width=0.19\textwidth]{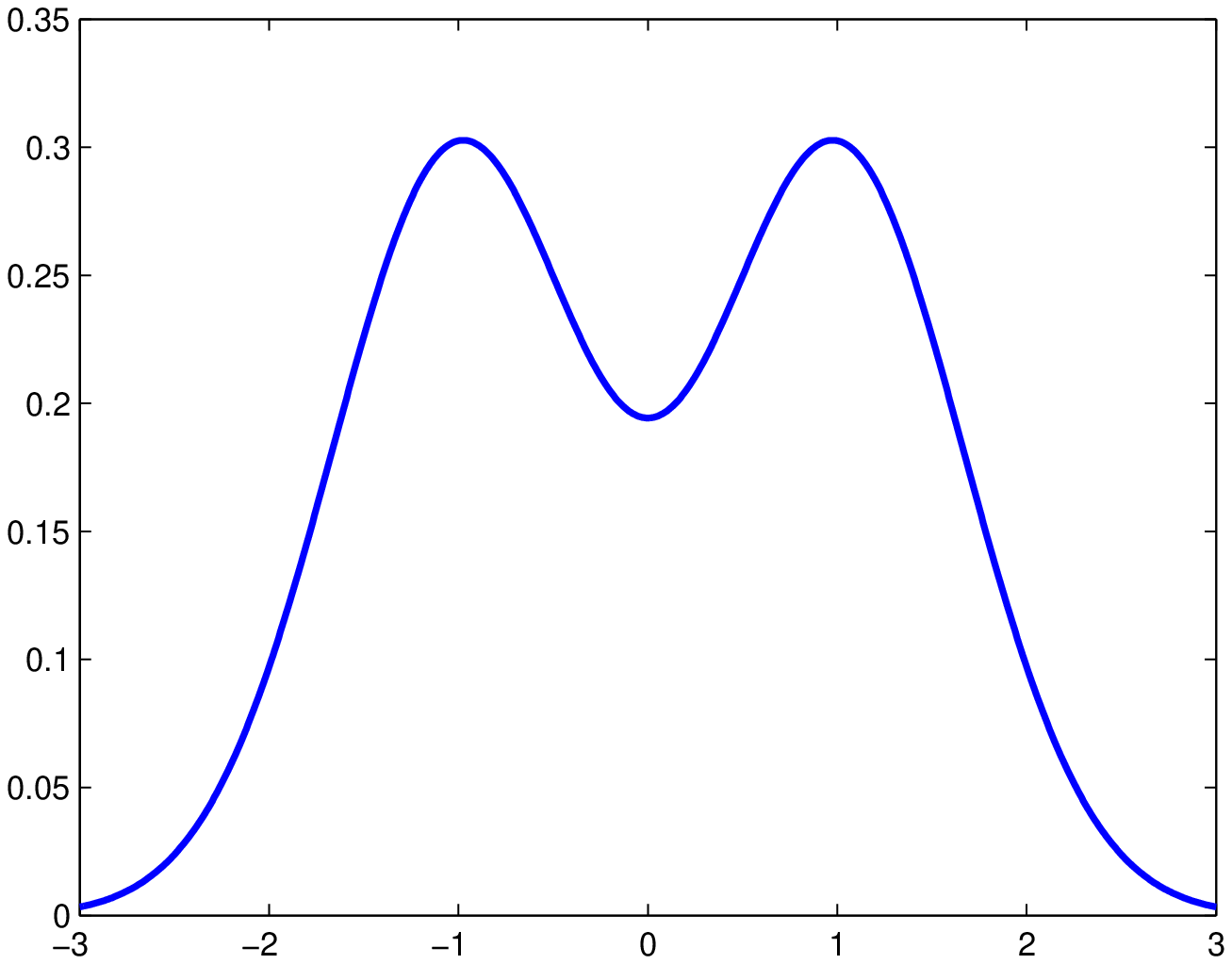}
\includegraphics[width=0.19\textwidth]{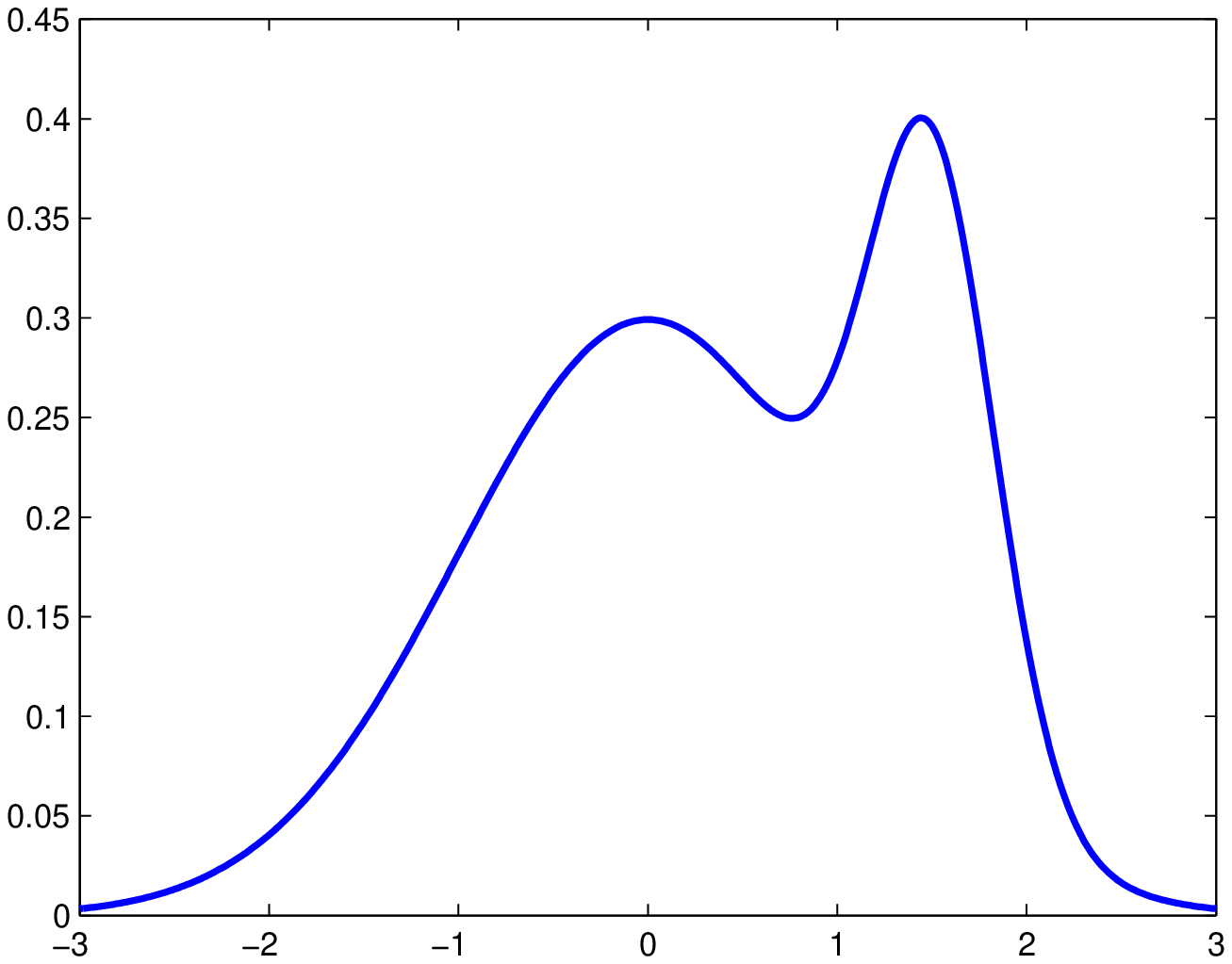}\includegraphics[width=0.19\textwidth]{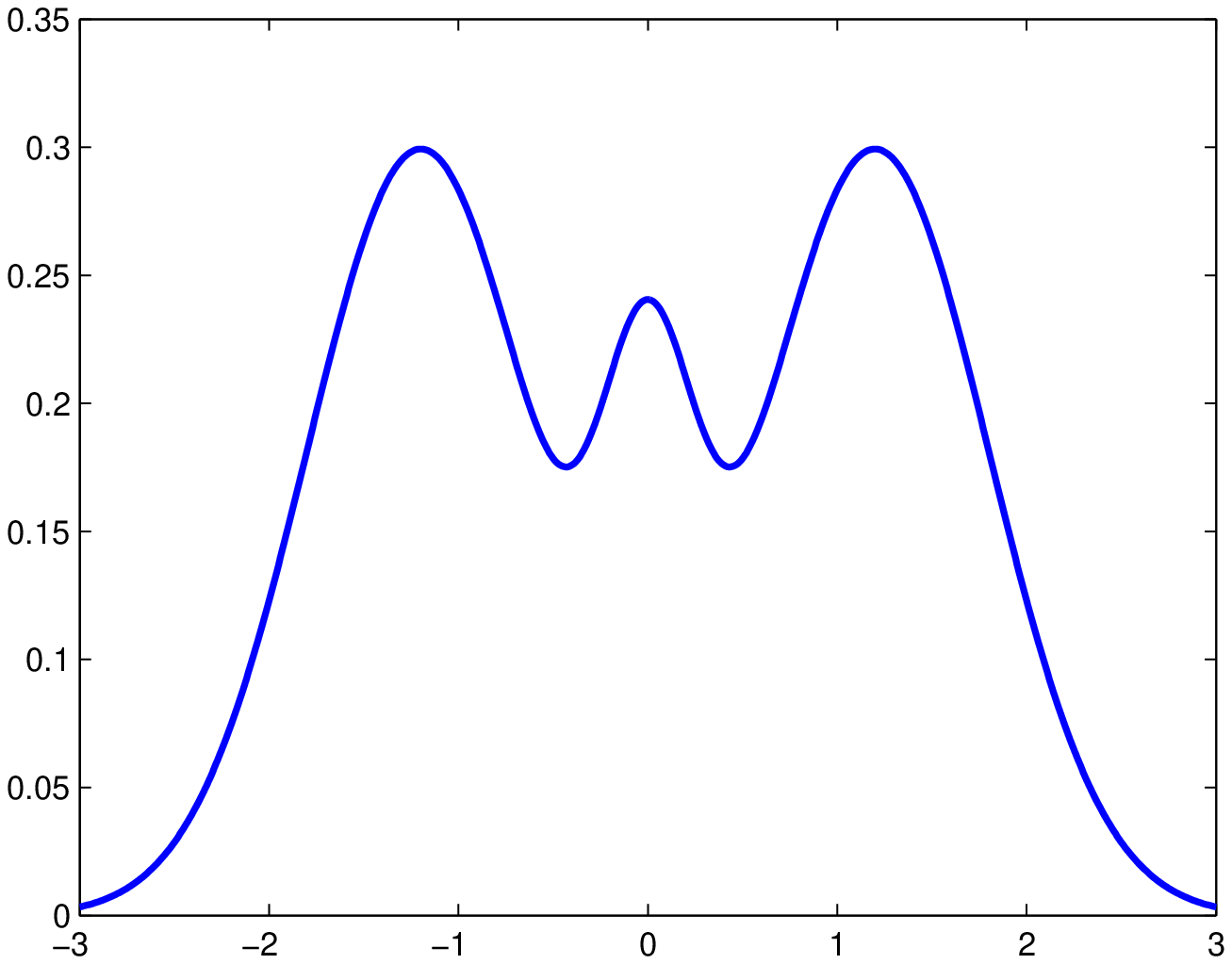}
\includegraphics[width=0.19\textwidth]{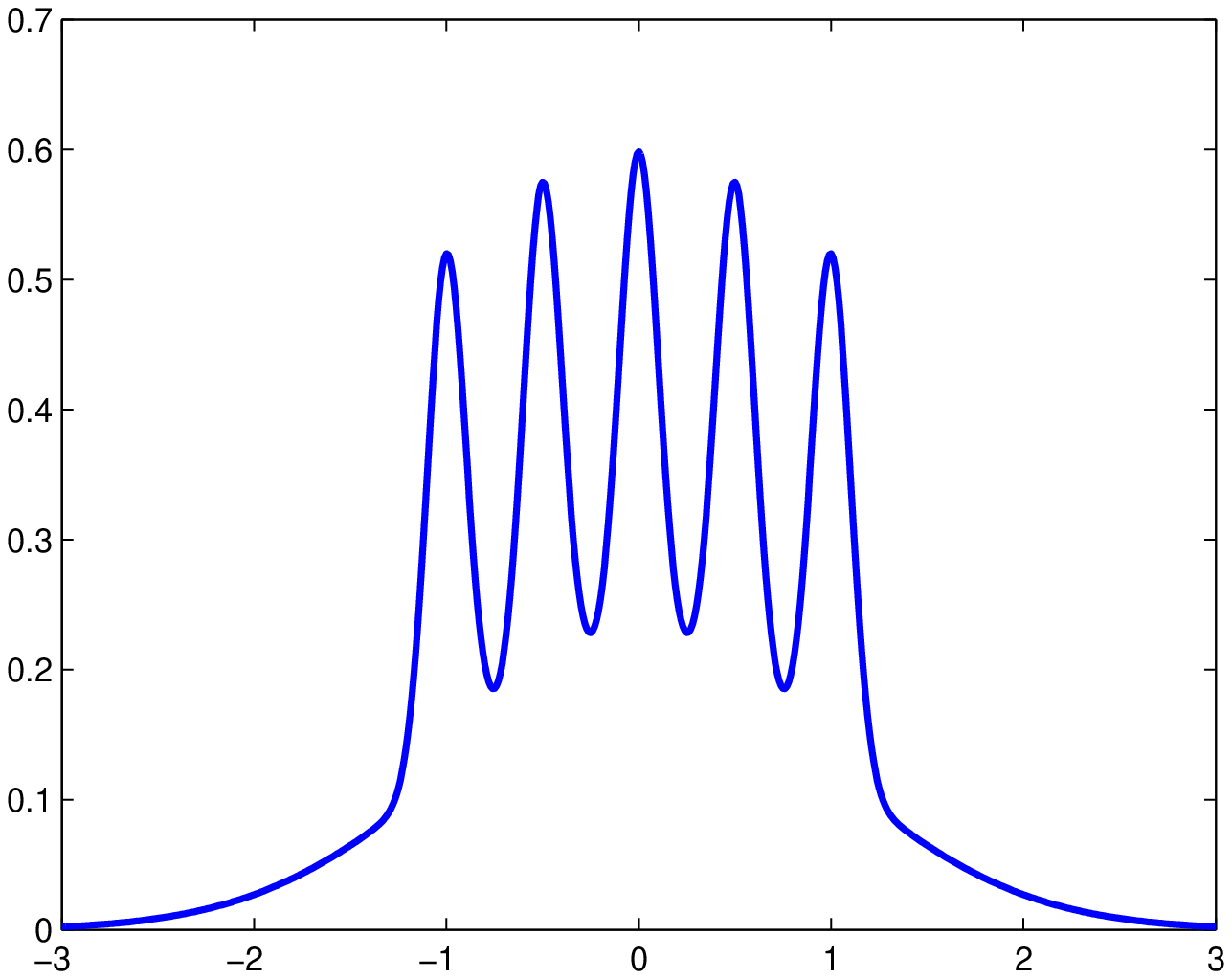}
\caption{Densities used to generate synthetic datasets. From left to right the densities are in the same order as in Table \ref{tab:dens}. The densities are taken from \cite{mclachlan00}}
\label{fig:dens}
\end{center}
\end{figure}

\begin{table}[!htbp]
\begin{center}
	\begin{tabular} {l|l|cccc}
		\hline
			Data & Method & AP & s & ESS & min(ESS)/s \\
		\hline
			\multirow{3}{*}{claw} & RMHMC & 0.80 & 2.54e+03 & (1524, 3474, 4586) & 0.60 \\ 
			 & RMLMC & 0.86 & 1.88e+03 & (2531, 4332, 5000) & 1.35 \\
			 & e-RMLMC & 0.82 & 1.46e+03 & (2436, 3455, 4608) & 1.67 \\
		\hline			 
			\multirow{3}{*}{trimodal} & RMHMC & 0.79 & 4.97e+02 & (4701, 4928, 5000) & 9.46 \\ 
			 & RMLMC & 0.82 & 2.01e+02 & (4978, 5000, 5000) & 24.77 \\
			 & e-RMLMC & 0.80 & 2.42e+02 & (4899, 4982, 5000) & 20.21 \\
		\hline
			\multirow{3}{*}{skewed} & RMHMC & 0.85 & 2.55e+02 & (5000, 5000, 5000) & 19.63 \\ 
			 & RMLMC & 0.82 & 1.13e+02 & (4698, 4940, 5000) & 41.68 \\
			 & e-RMLMC & 0.84 & 1.26e+02 & (4935, 5000, 5000) & 39.09 \\
		\hline
			\multirow{3}{*}{kurtotic} & RMHMC & 0.82 & 2.36e+02 & (5000, 5000, 5000) & 21.20 \\ 
			 & RMLMC & 0.85 & 1.27e+02 & (5000, 5000, 5000) & 39.34 \\
			 & e-RMLMC & 0.81 & 1.35e+02 & (5000, 5000, 5000) & 36.90 \\
		\hline
			\multirow{3}{*}{bimodal} & RMHMC & 0.86 & 2.69e+02 & (5000, 5000, 5000) & 18.56 \\ 
			 & RMLMC & 0.81 & 1.03e+02 & (4935, 4996, 5000) & 48.00 \\
			 & e-RMLMC & 0.85 & 1.08e+02 & (5000, 5000, 5000) & 46.43 \\
		\hline
	\end{tabular}
	\caption{Time, ESS and time-normalized ESS for Gaussian mixture models. Results are calculated on a 5,000 sample chain with a 5,000 sample burn-in session.}
	\label{MOG}
\end{center}
\end{table}

\section{Conclusions and Discussion}
Following the method of \cite{girolami11} for more efficient exploration of parameter space, we have proposed new sampling schemes to reduce the computational cost associated with using a position-specific mass matrix. To this end, we have developed a semi-explicit (RMLMC) integrator and a fully explicit (e-RMLMC) integrator for RMHMC and demonstrated their advantage in improving computational efficiency over the generalized leapfrog (RMHMC) method used by \cite{girolami11}. It is easy to show that for ${\bf G}({\boldsymbol\theta})= {\bf I}$, our method degenerates to standard HMC.

Future directions could involve splitting Hamiltonian \citep{dullweber87,sexton92,neal10,shahbabaSplitHMC} to develop explicit geometric integrators. For example, one could split a non-separable Hamiltonian dynamics into several smaller dynamics some of which can be analytically solved. A similar idea has been explored by \citep{chin09}, where the Hamiltonian, instead of the dynamic, is split.

Another possible research direction could be to approximate the mass matrix (Christofell Symbols). For many large-dimensional problems, the mass matrix could be appropriately approximated by a highly sparse matrix. This could further improve our method's computational efficiency.

\bibliographystyle{jasa}

\begin{thebibliography}{}

\bibitem[\protect\citeauthoryear{Amari and Nagaoka}{Amari and Nagaoka}{2006}]{amari00}
Amari, S. and H. Nagaoka (2000).
\newblock {\em {Methods of Information Geometry}\/}.
\newblock Oxford University Press.

\bibitem[\protect\citeauthoryear{Bishop and Goldberg}{Bishop and Goldberg}{1980}]{bishop80}
Bishop, Richard L. and Samuel I. Goldberg (2000).
\newblock {\em {Tensor Analysis on Manifolds}\/}.
\newblock Dover Publications, Inc.

\bibitem[\protect\citeauthoryear{Chin}{Chin}{2009}]{chin09}
Chin, S.~A. (2009).
\newblock {Explicit symplectic integrators for solving nonseparable Hamiltonians}.
\newblock {\em Physical Review E\/}~{\em 80\/}, 037701.

\bibitem[\protect\citeauthoryear{Ole F. Christensen, Gareth O. Roberts, Jeffery S. Rosenthal}{Ole F.~Christensen et~al}{2005}]{christensen05}
Christensen, O.F., G.O.~Roberts, J.S.~Rosenthal (2005, April).
\newblock {Scaling Limits for the Transient Phase of Local Metropolis-Hastings Algorithms}.
\newblock {\em Journal of the Royal Statistical Society: Series B\/}~{\em 67\/}(2), 253--268.

\bibitem[\protect\citeauthoryear{Duane, Kennedy, Pendleton, and Roweth}{Duane et~al.}{1987}]{duane87}
Duane, S., A.~D. Kennedy, B.~J. Pendleton, and D.~Roweth (1987).
\newblock Hybrid monte carlo.
\newblock {\em Physics Letters B\/}~{\em 195\/}(2), 216 -- 222.

\bibitem[\protect\citeauthoryear{Dullweber, Leimkuhler, McLachlan}{Dullweber et~al.}{1987}]{dullweber87}
Dullweber, A., B.~Leimkuhler, R.~McLachlan (1997).
\newblock Split-Hamiltonian Methods for Rigid Body Molecular Dynamics.
\newblock {\em Journal of Chemical Physics\/}~{\em 107\/}(15), 5840 -- 5852.

\bibitem[\protect\citeauthoryear{Geyer}{Geyer}{1992}]{geyer92}
Geyer, C.~J. (1992).
\newblock Practical markov chain monte carlo.
\newblock {\em Statistical Science\/}~{\em 7\/}(4), 473--483.

\bibitem[\protect\citeauthoryear{Girolami and Calderhead}{Girolami and Calderhead}{2011}]{girolami11}
Girolami, M. and Calderhead, B. (2011),
\newblock{{Riemann manifold Langevin and Hamiltonian Monte Carlo methods},}
\newblock{Journal of the Royal Statistical Society, Series B}, (with discussion) 73, 123--214.

\bibitem[\protect\citeauthoryear{Gelman and Hoffman}{Gelman and Hoffman}{2011}]{hoffman11}
Hoffman, M.~D. and A.~Gelman (2011).
\newblock The No-U-Turn Sampler: Adaptively Setting Path Lengths in Hamiltonian Monte Carlo.
\newblock arXiv:1111.4246v1.

\bibitem[\protect\citeauthoryear{Leimkuhler and Reich}{Leimkuhler and Reich}{2004}]{leimkuhler04}
Leimkuhler, B. and S.~Reich (2004).
\newblock {\em Simulating Hamiltonian Dynamics}.
\newblock Cambridge University Press.

\bibitem[\protect\citeauthoryear{Liu}{Liu}{2001}]{liu01}
Liu, J.~S. (2001).
\newblock {Molecular Dynamics and Hybrid Monte Carlo}.
\newblock In Liu, Jun S., {\em Monte Carlo Strategies in Scientific Computing}. Springer-Verlag.

\bibitem[\protect\citeauthoryear{Marin, Mengersen and Robert}{Marin et~al}{2005}]{marin05}
Marin, J.M., K.L.~Mengersen and C.~Robert (2005).
\newblock {Bayesian modeling and inference on mixtures of distributions}.
\newblock In D. Dey and C.R. Rao, {\em Handbook of Statistics}~{\em Volume 25}. Elsevier.

\bibitem[\protect\citeauthoryear{McLachlan and Peel}{McLachlan and Peel}{2000}]{mclachlan00}
McLachlan, G.J. and D.~Peel (2000).
\newblock {\em Finite Mixture Models}.
\newblock John Wiley $\&$ Sons, Inc., New York.

\bibitem[\protect\citeauthoryear{Neal}{Neal}{2010}]{neal10}
Neal, R.~M. (2010).
\newblock {MCMC using Hamiltonian dynamics}.
\newblock In S.~Brooks, A.~Gelman, G.~Jones, and X.~L. Meng (Eds.), {\em Handbook of Markov Chain Monte Carlo}. Chapman and Hall/CRC.

\bibitem[\protect\citeauthoryear{Sexton, and Weingarten}{Sexton and Weingarten}{1992}]{sexton92}
Sexton, J.~C. and D.H. Weingarten (1992).
\newblock {Hamiltonian evolution for the hybrid Monte Carlo algorithm}.
\newblock {\em Nuclear Physics B\/}~{\em 380\/}(3), 665--677.

\bibitem[\protect\citeauthoryear{Shahbaba, Lan, Johnson, and Neal}{Shahbaba et~al.}{2011}]{shahbabaSplitHMC}
Shahbaba, B., S.~Lan, W.~Johnson, and R.~Neal (2011).
\newblock Split hamiltonian monte carlo.
\newblock arXiv:1106.5941.

\bibitem[\protect\citeauthoryear{Stathopoulos, and Girolami}{Stathopoulos and Girolami}{2011}]{Stathopoulos11}
Stathopoulos, V. and Girolami, M.  (2011)
\newblock{Manifold MCMC for Mixtures}.
\newblock{in {\it K. Mengersen, C. P. Robert and M. D. Titteringhton, Mixture Estimation and Applications} }
\newblock{John Wiley \& Sons}


\bibitem[\protect\citeauthoryear{Verlet}{Verlet}{1967}]{verlet67}
Verlet, Loup (1967).
\newblock {Computer "Experiments" on Classical Fluids. I. Thermodynamical Properties of Lennard-Jones Molecules}.
\newblock {\em Physical Review \/}~{\em 159\/}, 98--103.

\end{thebibliography}

\newpage
\begin{center}
{\huge \bf Appendix: Derivations and Proofs}
\end{center}
\appendix
In what follows, we show the detailed derivations of our methods. We adopt Einstein notation (summation convention), so whenever the index appears twice in a mathematical expression, we sum over it: e.g., $a_{i}b^{i}:=\sum_{i}a_{i}b^{i}$, $\Gamma_{ij}^{k}v^{i}v^{j}:=\sum_{i,j}\Gamma_{ij}^{k}v^{i}v^{j}$. A lower index is used for the covariant tensor, whose components vary by the same transformation as the change of basis�e.g., gradient�whereas the upper index is reserved for the contravariant tensor, whose components vary in the opposite way as the change of basis in order to compensate: e.g. velocity vector. Interested readers should refer to \cite{bishop80}.

\section{Transformation of Hamiltonian Dynamics}\label{derRMLD}
To derive the dynamic \eqref{rmld} from the Hamiltonian dynamic \eqref{rmhd}, the first equation in \eqref{rmld} is directly obtained from the assumed transformation: $\dot {\boldsymbol\theta}^k = {\bf g}^{kl}{\bf p}_l={\bf v}^k$. For the second equation in \eqref{rmld}, we have
\[
\dot {\bf p}_{l} = \frac{d ({\bf g}_{lj}({\boldsymbol\theta}){\bf v}^{j})}{dt} =\frac{\pa {\bf g}_{lj}}{\pa {\boldsymbol\theta}_{i}}\dot {\boldsymbol\theta}^{i}{\bf v}^{j} + {\bf g}_{lj}\dot {\bf v}^{j} = \pa_{i} {\bf g}_{lj}{\bf v}^{i}{\bf v}^{j} + {\bf g}_{lj}\dot {\bf v}^{j}
\]
Further, from Equation \eqref{rmhd} we have
\[
\begin{split}
\dot {\bf p}_{l} & = -\pa_{l} \phi({\boldsymbol\theta}) + \frac{1}{2}{\bf v}^{\textsf T}\pa_{l} {\bf G}({\boldsymbol\theta}) {\bf v} = -\pa_{l} \phi + \frac{1}{2}{\bf g}_{ij,l}{\bf v}^{i}{\bf v}^{j}\\
& = \pa_{i} {\bf g}_{lj}{\bf v}^{i}{\bf v}^{j} + {\bf g}_{lj}\dot {\bf v}^{j}
\end{split}
\]
which means
\[
{\bf g}_{lj}\dot {\bf v}^{j} = -(\pa_{i} {\bf g}_{lj} - \frac{1}{2}\pa_{l} {\bf g}_{ij}){\bf v}^{i}{\bf v}^{j} -\pa_{l} \phi
\]
By multiplying ${\bf G}^{-1}= ({\bf g}^{kl})$ on both sides, we have
\begin{equation}\label{dvij}
\dot {\bf v}^{k} = \delta^{k}_{j}\dot {\bf v}^{j} = -{\bf g}^{kl}(\pa_{i} {\bf g}_{lj} - \frac{1}{2}\pa_{l} {\bf g}_{ij}){\bf v}^{i}{\bf v}^{j} -{\bf g}^{kl}\pa_{l} \phi
\end{equation}
Since $i, j$ are symmetric in the first summand, switching them gives the following equations:
\begin{equation}\label{dvji}
\dot {\bf v}^{k} = -{\bf g}^{kl}(\pa_{j} {\bf g}_{li} - \frac{1}{2}\pa_{l} {\bf g}_{ji}){\bf v}^{i}{\bf v}^{j} -{\bf g}^{kl}\pa_{l} \phi
\end{equation}
which in turn gives the final form of Equation \eqref{rmld} after adding equations \eqref{dvij} and \eqref{dvji} and dividing the results by two:
\begin{equation*}
\dot {\bf v}^k = -\Gamma^k_{ij}({\boldsymbol\theta}){\bf v}^i {\bf v}^j - {\bf g}^{kl}({\boldsymbol\theta})\pa_{l} \phi({\boldsymbol\theta})
\end{equation*}
Here, $\Gamma^k_{ij}({\boldsymbol\theta}) := \frac12 {\bf g}^{kl}(\pa_{i} {\bf g}_{lj} + \pa_{j} {\bf g}_{il} - \pa_{l} {\bf g}_{ij})$ is Christoffel Symbol of second kind.

Note that the new dynamic \eqref{rmld} still preserves the original Hamiltonian $H({\boldsymbol\theta}, {\bf p}={\bf G}({\boldsymbol\theta}){\bf v})$. This is of course intuitive, but it also can be proven as follows:
\[
\begin{split}
\frac{d}{dt} H({\boldsymbol\theta}, {\bf G}({\boldsymbol\theta}){\bf v}) & = \dot {\boldsymbol\theta}^{\textsf T} \frac{\pa }{\pa {\boldsymbol\theta}} H({\boldsymbol\theta}, {\bf G}({\boldsymbol\theta}){\bf v}) + \dot {\bf v}^{\textsf T} \frac{\pa }{\pa {\bf v}} H({\boldsymbol\theta}, {\bf G}({\boldsymbol\theta}){\bf v})\\
& = {\bf v}^{\textsf T}\left[ \nabla_{\boldsymbol\theta} \phi({\boldsymbol\theta}) + \frac12 {\bf v}^{\textsf T}\pa {\bf G}({\boldsymbol\theta}) {\bf v}\right] + \left[-{\bf v}^{\textsf T}\Gamma({\boldsymbol\theta}) {\bf v} - {\bf G}({\boldsymbol\theta})^{-1}\nabla_{\boldsymbol\theta} \phi({\boldsymbol\theta})\right]^{\textsf T} {\bf G}({\boldsymbol\theta}){\bf v}\\
& = {\bf v}^{\textsf T}\nabla_{\boldsymbol\theta} \phi({\boldsymbol\theta}) - \left(\nabla_{\boldsymbol\theta} \phi({\boldsymbol\theta})\right)^{\textsf T}{\bf v} + \frac12{\bf v}^{\textsf T}\left( {\bf v}^{\textsf T}\pa {\bf G}({\boldsymbol\theta}) {\bf v}\right)- ({\bf v}^{\textsf T}\tilde\Gamma({\boldsymbol\theta}) {\bf v})^{\textsf T} {\bf v}\\
& = 0+ 0 =0
\end{split}
\]
where ${\bf v}^{\textsf T}\Gamma({\boldsymbol\theta}) {\bf v}$ is a vector whose $k$th element is  $\Gamma^k_{ij}({\boldsymbol\theta}){\bf v}^i {\bf v}^j$. The second 0 is due to the triple form $({\bf v}^{\textsf T}\tilde\Gamma({\boldsymbol\theta}) {\bf v})^{\textsf T} {\bf v}= \tilde\Gamma_{ijk}{\bf v}^{i}{\bf v}^{j}{\bf v}^{k}=\frac12 \pa_{k} {\bf g}_{ij}{\bf v}^{i}{\bf v}^{j}{\bf v}^{k}$, where $\tilde \Gamma$ is Christoffel Symbol of first kind with elements $\tilde\Gamma_{ijk}({\boldsymbol\theta}) := {\bf g}_{kl}\Gamma_{ij}^{l}({\boldsymbol\theta}) = \frac12 (\pa_{i} {\bf g}_{kj} + \pa_{j} {\bf g}_{ik} - \pa_{k} {\bf g}_{ij})$.

\section{Derivation of semi-explicit Riemannian Manifold Lagrangian Monte Carlo (RMLMC)}\label{semi-RMLMC} 
Consider the following generalized leapfrog integration scheme:
\begin{eqnarray*}
{\bf p}^{(n+1/2)}  & = & {\bf p}^{(n)} - \frac{\eps}{2}\frac{\pa H}{\pa {\boldsymbol\theta}}({\boldsymbol\theta}^{(n)},{\bf p}^{(n+1/2)})\\
{\boldsymbol\theta}^{(n+1)} & = & {\boldsymbol\theta}^{(n)} + \frac{\eps}{2}\left[ \frac{\pa H}{\pa {\bf p}}({\boldsymbol\theta}^{(n)},{\bf p}^{(n+1/2)}) + \frac{\pa H}{\pa {\bf p}}({\boldsymbol\theta}^{(n+1)},{\bf p}^{(n+1/2)}) \right] \\
{\bf p}^{(n+1)} & = & {\bf p}^{(n+1/2)} - \frac{\eps}{2}\frac{\pa H}{\pa {\boldsymbol\theta}}({\boldsymbol\theta}^{(n)},{\bf p}^{(n+1/2)})
\end{eqnarray*}
We composite implicit steps for velocity and explicit step for position within a leapfrog step to integrate dynamic \eqref{rmld} and derive the following semi-explicit integrator:
\begin{eqnarray}\label{GLv}
{\bf v}^{(n+1/2)}  & = & {\bf v}^{(n)} - \frac{\eps}{2}[({\bf v}^{(n+1/2)})^{\textsf T} \Gamma({\boldsymbol\theta}^{(n)}) {\bf v}^{(n+1/2)} + {\bf G}({\boldsymbol\theta}^{(n)})^{-1} \nabla_{\boldsymbol\theta}\phi({\boldsymbol\theta}^{(n)})] \label{GLv:1}\\
{\boldsymbol\theta}^{(n+1)} & = & {\boldsymbol\theta}^{(n)} + \eps {\bf v}^{(n+1/2)}  \label{GLv:2}\\
{\bf v}^{(n+1)} & = & {\bf v}^{(n+1/2)} - \frac{\eps}{2}[({\bf v}^{(n+1/2)})^{\textsf T} \Gamma({\boldsymbol\theta}^{(n+1)}) {\bf v}^{(n+1/2)} + {\bf G}({\boldsymbol\theta}^{(n+1)})^{-1} \nabla_{\boldsymbol\theta}\phi({\boldsymbol\theta}^{(n+1)})] \label{GLv:3}
\end{eqnarray}
The time-reversibility of this integrator can be shown by switching $({\boldsymbol\theta},{\bf v})^{(n+1)}$ and $({\boldsymbol\theta},{\bf v})^{(n)}$ and negating velocity. The resulting integrator, however, is no longer volume-preserving (see subsection \ref{detadj1}). Nevertheless, based on proposition \ref{dbwda}, we can still have detailed balance after determinant adjustment. \citep[See ][for more details.]{liu01}
\begin{prop}[Detailed Balance Condition with determinant adjustment]\label{dbwda}
Denote ${\bf z}=({\boldsymbol\theta},{\bf p})$, ${\bf z}'=\hat T_{L}({\bf z})$ for some time reversible integrator $\hat T_{L}$ to the Lagrangian dynamic. If the acceptance probability is adjusted in the following way:
\begin{equation*}
\tilde\alpha({\bf z},{\bf z'}) =\min\left\{1,\frac{\exp(-H({\bf z'}))}{\exp(-H({\bf z}))}|\det \hat T_{L}|\right\}
\end{equation*}
then the detailed balance condition still holds
\begin{equation*}
\tilde\alpha({\bf z},{\bf z'}) \mathbb P(d{\bf z}) = \tilde\alpha({\bf z'},{\bf z}) \mathbb P(d{\bf z'})
\end{equation*}
\end{prop}
\begin{proof}
\[
\begin{split}
\tilde\alpha({\bf z},{\bf z'}) \mathbb P(d{\bf z}) & =\min\left\{1,\frac{\exp(-H({\bf z'}))}{\exp(-H({\bf z}))}\left|\frac{d{\bf z'}}{d{\bf z}}\right|\right\} \exp(-H({\bf z}))d{\bf z}\\
& \overset{{\bf z}= \hat T_{L}^{-1}({\bf z'})}= \min\left\{\exp(-H({\bf z})),\exp(-H({\bf z'}))\left|\frac{d{\bf z'}}{d{\bf z}}\right|\right\} \left|\frac{d{\bf z}}{d{\bf z'}}\right|d{\bf z'}\\
& = \min\left\{1,\frac{\exp(-H({\bf z}))}{\exp(-H({\bf z'}))}\left|\frac{d{\bf z}}{d{\bf z'}}\right|\right\} \exp(-H({\bf z'})) {\bf z'} = \tilde\alpha({\bf z'},{\bf z}) \mathbb P(d{\bf z'})
\end{split}
\]
\qed
\end{proof}
Therefore, the acceptance probability could be calculated based on $H({\boldsymbol\theta},{\bf G}({\boldsymbol\theta}){\bf v})$. However, it also could be also calculated as follows based on the energy function ${\bf E}({\boldsymbol\theta},{\bf v})$ defined in section \ref{useV},
\begin{equation}\label{energyv}
{\bf E}({\boldsymbol\theta}, {\bf v}) = U({\boldsymbol\theta}) + K({\boldsymbol\theta}, {\bf v}) = -\log p({\boldsymbol\theta}) -\frac12 \log\det {\bf G}({\boldsymbol\theta}) + \frac12 {\bf v}^{\textsf T}{\bf G}({\boldsymbol\theta}){\bf v}
\end{equation}
To show their equivalence, we note that $\displaystyle \left|\frac{\pa({\boldsymbol\theta}',{\bf p}')}{\pa({\boldsymbol\theta},{\bf p})}\right| = \frac{\det({\bf G}({\boldsymbol\theta}'))}{\det({\bf G}({\boldsymbol\theta}))} \left|\frac{\pa({\boldsymbol\theta}',{\bf v}')}{\pa({\boldsymbol\theta},{\bf v})}\right|$ and proved our claim as follows:
\[
\begin{split}
\tilde\alpha & = \min\left\{1,\frac{\exp(-H({\boldsymbol\theta}',{\bf p}'))}{\exp(-H({\boldsymbol\theta},{\bf p}))}  \left|\frac{\pa({\boldsymbol\theta}',{\bf p}')}{\pa({\boldsymbol\theta},{\bf p})}\right| \right\} = \min\left\{1,\frac{\exp(-H({\boldsymbol\theta}',{\bf G}({\boldsymbol\theta}'){\bf v}'))}{\exp(-H({\boldsymbol\theta},{\bf G}({\boldsymbol\theta}){\bf v}))}  \left|\frac{\pa({\boldsymbol\theta}',{\bf p}')}{\pa({\boldsymbol\theta},{\bf p})}\right| \right\} \\
& = \min\left\{1,\frac{\exp\{-(\log p({\boldsymbol\theta}') +\frac12 \log\det {\bf G}({\boldsymbol\theta}') + 
\frac12 {\bf v'}^{\textsf T}{\bf G}({\boldsymbol\theta}'){\bf v'})\}}{\exp\{-(\log p({\boldsymbol\theta}) +\frac12 \log\det {\bf G}({\boldsymbol\theta}) + 
\frac12 {\bf v}^{\textsf T}{\bf G}({\boldsymbol\theta}){\bf v})\}} \frac{\det({\bf G}({\boldsymbol\theta}'))}{\det({\bf G}({\boldsymbol\theta}))} \left|\frac{\pa({\boldsymbol\theta}',{\bf v}')}{\pa({\boldsymbol\theta},{\bf v})}\right| \right\}\\
& = \min\left\{1,\frac{\exp\{-(\log p({\boldsymbol\theta}') -\frac12 \log\det {\bf G}({\boldsymbol\theta}') + 
\frac12 {\bf v'}^{\textsf T}{\bf G}({\boldsymbol\theta}'){\bf v'})\}}{\exp\{-(\log p({\boldsymbol\theta}) -\frac12 \log\det {\bf G}({\boldsymbol\theta}) + 
\frac12 {\bf v}^{\textsf T}{\bf G}({\boldsymbol\theta}){\bf v})\}} \left|\frac{\pa({\boldsymbol\theta}',{\bf v}')}{\pa({\boldsymbol\theta},{\bf v})}\right| \right\}\\
& = \min\left\{1,\frac{\exp(-{\bf E}({\boldsymbol\theta}',{\bf v}'))}{\exp(-\bf{E}({\boldsymbol\theta},{\bf v}))}  \left|\frac{\pa({\boldsymbol\theta}',{\bf v}')}{\pa({\boldsymbol\theta},{\bf v})}\right| \right\}
\end{split}
\]

\subsection{Volume Correction}\label{detadj1}
To adjust volume, we must derive the Jacobian determinant, $\det {\bf J}:=\left|\frac{\pa({\boldsymbol\theta}^{(L+1)},{\bf v}^{(L+1)})}{\pa({\boldsymbol\theta}^{(1)},{\bf v}^{(1)})}\right|$ using wedge products.
\begin{dfn}[Differential Forms, Wedge Product]
Differential one-form $\alpha: TM^{D}\to \mathbb R$ on a differential manifold $M^{D}$ is a smooth mapping from tangent space $TM^{D}$ to $\mathbb R$, which can be expressed as linear combination of differentials of local coordinates: $\alpha = f_{i}dx^{i}=: f\cdot dx$.

For example, if $f: \mathbb R^{D}\to \mathbb R$ is a smooth function, then its directional derivative along a vector $v\in \mathbb R^{D}$, denoted by $df(v)$ is given by
\[
df(v) = \frac{\pa f}{\pa z_{i}} v^{i}
\]
then $df(\cdot)$ is a linear functional of $v$, called the differential of $f$ at $z$ and is an example of a differential one-form.
In particular, $dz^{i}(v)=v^{i}$, thus
\[
df(v) = \frac{\pa f}{\pa z_{i}} dz^{i}(v), \quad then \; df = \frac{\pa f}{\pa z_{i}} dz^{i}
\]

Wedge Product of two one-form $\alpha, \beta$ is a 2-form $\alpha\wedge\beta$ anti-symmetric bilinear function on tangent space which has the following properties ($\alpha,\beta,\gamma$ one-forms, $A$ be a square matrix of same dimension $D$):
 \begin{itemize}
 \item $\alpha\wedge\alpha=0$
 \item $\alpha\wedge(\beta+\gamma) = \alpha\wedge\beta+\alpha\wedge\gamma$ (thus $\alpha\wedge\beta=-\beta\wedge\alpha$)
 \item $\alpha\wedge A\beta=A^{\textsf T}\alpha\wedge\beta$
 \end{itemize}
\end{dfn}

The following proposition enables us to calculate the Jacobian determinant denoted as $\det {\bf J}$.
\begin{prop}\label{wedge}
Let $T_{L}: ({\boldsymbol\theta}^{(1)},{\bf v}^{(1)})\to ({\boldsymbol\theta}^{(L+1)},{\bf v}^{(L+1)})$ be evolution of a smooth flow, then
\[
d{\boldsymbol\theta}^{(L+1)}\wedge d{\bf v}^{(L+1)} = \frac{\pa({\boldsymbol\theta}^{(L+1)},{\bf v}^{(L+1)})}{\pa({\boldsymbol\theta}^{(1)},{\bf v}^{(1)})} d{\boldsymbol\theta}^{(1)}\wedge d{\bf v}^{(1)}
\]
\end{prop}
Note that the Jacobian determinant $\det {\bf J}$ can also be regarded as Radon-Nikodym derivative of two probability measures: $\det {\bf J} = \displaystyle \frac{\mathbb P(d{\boldsymbol\theta}^{(L+1)},d{\bf v}^{(L+1)})}{\mathbb P(d{\boldsymbol\theta}^{(1)}, d{\bf v}^{(1)})}$, where $\mathbb P(d{\boldsymbol\theta},d{\bf v})=p({\boldsymbol\theta},{\bf v})d{\boldsymbol\theta}d{\bf v}$. We have
\begin{eqnarray*}
d{\bf v}^{(n+1/2)}  & = & d{\bf v}^{(n)} - \eps ({\bf v}^{(n+1/2)})^{\textsf T} \Gamma({\boldsymbol\theta}^{(n)}) d{\bf v}^{(n+1/2)}+ (**)d{\boldsymbol\theta}^{(n)}\\
d{\boldsymbol\theta}^{(n+1)} & = & d{\boldsymbol\theta}^{(n)} + \eps d{\bf v}^{(n+1/2)}\\
d{\bf v}^{(n+1)} & = & d{\bf v}^{(n+1/2)} - \eps ({\bf v}^{(n+1/2)})^{\textsf T} \Gamma({\boldsymbol\theta}^{(n+1)}) d{\bf v}^{(n+1/2)} +(**)d{\boldsymbol\theta}^{(n+1)}
\end{eqnarray*}
where ${\bf v}^{\textsf T} \Gamma({\boldsymbol\theta})$ is a matrix whose $(k,j)$th element is ${\bf v}^{i}\Gamma_{ij}^{k}({\boldsymbol\theta})$. Therefore,
\[
\begin{split}
d{\boldsymbol\theta}^{(n+1)} \wedge d{\bf v}^{(n+1)} & = [{\bf I}-\eps ({\bf v}^{(n+1/2)})^{\textsf T} \Gamma({\boldsymbol\theta}^{(n+1)})]^{\textsf T} d{\boldsymbol\theta}^{(n+1)} \wedge d{\bf v}^{(n+1/2)}\\
& = [{\bf I}-\eps ({\bf v}^{(n+1/2)})^{\textsf T} \Gamma({\boldsymbol\theta}^{(n+1)})]^{\textsf T} d{\boldsymbol\theta}^{(n)} \wedge d{\bf v}^{(n+1/2)}\\
& = [{\bf I}-\eps ({\bf v}^{(n+1/2)})^{\textsf T} \Gamma({\boldsymbol\theta}^{(n+1)})]^{\textsf T} [{\bf I}+\eps ({\bf v}^{(n+1/2)})^{\textsf T} \Gamma({\boldsymbol\theta}^{(n)})]^{-T} d{\boldsymbol\theta}^{(n)} \wedge d{\bf v}^{(n)}
\end{split}
\]
For volume adjustment, we must use the following Jacobian determinant accumulated along leap frog steps:
\begin{equation}\label{detglv}
\det {\bf J}_{LMC} := \left|\frac{\pa({\boldsymbol\theta}^{(L+1)},{\bf v}^{(L+1)})}{\pa({\boldsymbol\theta}^{(1)},{\bf v}^{(1)})}\right| = \prod_{n=1}^{L} \frac{\det(I-\eps ({\bf v}^{(n+1/2)})^{\textsf T} \Gamma({\boldsymbol\theta}^{(n+1)}))}{\det(I+\eps ({\bf v}^{(n+1/2)})^{\textsf T} \Gamma({\boldsymbol\theta}^{(n)}))}\\[12pt]
\end{equation}
As a result, the acceptance probability becomes
\begin{equation*}
\alpha_{LMC} = \min\{1,\exp(-{\bf E}({\boldsymbol\theta}^{(L+1)},{\bf v}^{(L+1)})+{\bf E}({\boldsymbol\theta}^{(1)},{\bf v}^{(1)}))|\det {\bf J}_{LMC}|\}
\end{equation*}
Using this acceptance probability, we are able to derive a semi-explicit integrator for RMLMC as shown in Algorithm \ref{Alg:RMLMC}. In this approach, the updates for ${\boldsymbol\theta}$ are explicit, while updating ${\bf v}$ remains implicit.

\section{Connection to Lagrangian Dynamics}\label{connection2L}
We now show that the above dynamic \eqref{rmld} is indeed \emph{Lagrangian} dynamic. We define {Lagrangian} as follows:
\begin{equation*}
{\bf L} = \frac12 {\bf v}^{\textsf T}{\bf G}({\boldsymbol\theta}){\bf v} - \phi({\boldsymbol\theta}) \;
\end{equation*}
Using variation calculus to minimize the Lagrangian, we obtain a Euler-Lagrange equation of the second kind,
\begin{equation*}
\frac{\pa {\bf L}}{\pa {\boldsymbol\theta}}=\frac{d}{dt} \frac{\pa {\bf L}}{\pa \dot{\boldsymbol\theta}} \;
\end{equation*}
which is
\begin{equation}\label{variation}
\ddot {\boldsymbol\theta}  =  - \dot{\boldsymbol\theta}^{\textsf T} \Gamma({\boldsymbol\theta}) {\bf v} - {\bf G}({\boldsymbol\theta})^{-1} \nabla_{\boldsymbol\theta} \phi({\boldsymbol\theta})
\end{equation}
This is equivalent to the new dynamic \eqref{rmld} by taking the time derivative on the first equation and equating it with the time derivative on the second.

\section{Derivation of explicit Riemannian Manifold Lagrangian Monte Carlo (e-RMLMC)}\label{xplct-RMLMC}
We now propose an additional modification of Algorithm \ref{Alg:RMLMC} to resolve the remaining implicit equation \eqref{GLv:1}, while keeping time-reversibility to ensure the ergodicity of the induced Markov chain. We do this by modifying the symmetric quadratic form in equations \eqref{GLv:1}, 
\begin{eqnarray}\label{TSv}
{\bf v}^{(n+1/2)} & = & {\bf v}^{(n)} - \frac{\eps}{2} [({\bf v}^{(n)})^{\textsf T}\Gamma({\boldsymbol\theta}^{(n)}) {\bf v}^{(n+1/2)} + {\bf G}({\boldsymbol\theta}^{(n)})^{-1}\nabla_{\boldsymbol\theta}\phi({\boldsymbol\theta}^{(n)})] \label{TSv:1}\\
{\boldsymbol\theta}^{(n+1)} & = & {\boldsymbol\theta}^{(n)} + \eps {\bf v}^{(n+1/2)} \label{TSv:2}\\
{\bf v}^{(n+1)} & = & {\bf v}^{(n+1/2)} - \frac{\eps}{2} [({\bf v}^{(n+1/2)})^{\textsf T}\Gamma({\boldsymbol\theta}^{(n+1)}) {\bf v}^{(n+1)} + {\bf G}({\boldsymbol\theta}^{(n+1)})^{-1}\nabla_{\boldsymbol\theta}\phi({\boldsymbol\theta}^{(n+1)})] \label{TSv:3}
\end{eqnarray}
The resulting integrator is completely explicit since both updates of velocity \eqref{TSv:1} and \eqref{TSv:3} can be solved as follows:
\begin{eqnarray}\label{TSvsolved}
{\bf v}^{(n+1/2)} & = & [{\bf I} + \frac{\eps}{2} ({\bf v}^{(n)})^{\textsf T}\Gamma({\boldsymbol\theta}^{(n)})]^{-1}
[{\bf v}^{(n)} - \frac{\eps}{2} {\bf G}({\boldsymbol\theta}^{(n)})^{-1}
\nabla_{\boldsymbol\theta}\phi({\boldsymbol\theta}^{(n)})]  \label{TSvsolved:1}\\
{\bf v}^{(n+1)} & = & [{\bf I} + \frac{\eps}{2} ({\bf v}^{(n+1/2)})^{\textsf T}\Gamma({\boldsymbol\theta}^{(n+1)})]^{-1}[{\bf v}^{(n+1/2)} - \frac{\eps}{2} {\bf G}({\boldsymbol\theta}^{(n+1)})^{-1}\nabla_{\boldsymbol\theta}\phi({\boldsymbol\theta}^{(n+1)})] \label{TSvsolved:2}
\end{eqnarray}
\subsection{Convergence of Numerical Solution}\label{convg}
We now show that the discretization error $e_n=\Vert {\bf z}(t_n) - {\bf z}^{(n)}\Vert =\Vert ({\boldsymbol\theta}(t_n),{\bf v}(t_n)) - ({\boldsymbol\theta}^{(n)},{\bf v}^{(n)})\Vert$ (i.e. the difference between the true solution and the numerical solution) accumulated over final time interval $[0,T]$, is bounded and goes to zeros as the stepsize $\eps$ goes to zero. (See \cite{leimkuhler04} for a similar proof for the generalized leapfrog method.) Here, we assume that $F({\boldsymbol\theta},{\bf v}):= {\bf v}^{\textsf T}\Gamma({\boldsymbol\theta}) {\bf v} + {\bf G}({\boldsymbol\theta})^{-1} \nabla_{\boldsymbol\theta}\phi({\boldsymbol\theta})$ is smooth; hence, $F$ and its derivatives are uniformly bounded as $({\boldsymbol\theta},{\bf v})$ evolves within finite time duration $T$.
We expand the true solution ${\bf z}(t_{n+1})$ at $t_n$:
\[
\begin{split}
{\bf z}(t_{n+1}) & = {\bf z}(t_n)+\dot {\bf z}(t_n)\eps +\frac{1}{2}\ddot {\bf z}(t_n)\eps^2 +o(\eps)\\
& = \begin{bmatrix}{\boldsymbol\theta}(t_n)\\{\bf v}(t_n)\end{bmatrix} + \begin{bmatrix}{\bf v}(t_n)\\-F({\boldsymbol\theta}(t_n),{\bf v}(t_n))\end{bmatrix}\eps + \frac12 \begin{bmatrix}-F({\boldsymbol\theta}(t_n),{\bf v}(t_n))\\-\frac{\pa F}{\pa {\boldsymbol\theta}}{\bf v}(t_n)+\frac{\pa F}{\pa {\bf v}}F({\boldsymbol\theta}(t_n),{\bf v}(t_n))\end{bmatrix}\eps^2 +o(\eps)\\
& = \begin{bmatrix}{\boldsymbol\theta}(t_n)\\{\bf v}(t_n)\end{bmatrix} + \begin{bmatrix}{\bf v}(t_n)\\-F({\boldsymbol\theta}(t_n),{\bf v}(t_n))\end{bmatrix}\eps + O(\eps^2)
\end{split}
\]
Next, we simplify the expression of the numerical solutions ${\bf z}^{(n+1)} = \begin{bmatrix}{\boldsymbol\theta}^{(n+1)}\\{\bf v}^{(n+1)}\end{bmatrix}$ for the fully explicit integrator and compare it to the above true solutions. To this end, we rewrite equation \eqref{TSvsolved:1} as follows:
\[
\begin{split}
{\bf v}^{(n+1/2)} & = [{\bf I} + \frac{\eps}{2} ({\bf v}^{(n)})^{\textsf T}\Gamma({\boldsymbol\theta}^{(n)})]^{-1}[{\bf v}^{(n)} - \frac{\eps}{2} {\bf G}({\boldsymbol\theta}^{(n)})^{-1}\nabla_{\boldsymbol\theta}\phi({\boldsymbol\theta}^{(n)})]\\
& = {\bf v}^{(n)} - [{\bf I} + \frac{\eps}{2} ({\bf v}^{(n)})^{\textsf T}\Gamma({\boldsymbol\theta}^{(n)})]^{-1}[({\bf v}^{(n)})^{\textsf T}\Gamma({\boldsymbol\theta}^{(n)}) {\bf v}^{(n)} + \frac{\eps}{2} {\bf G}({\boldsymbol\theta}^{(n)})^{-1}\nabla_{\boldsymbol\theta}\phi({\boldsymbol\theta}^{(n)})]\\
& = {\bf v}^{(n)} - [{\bf I} +  \frac{\eps}{2} ({\bf v}^{(n)})^{\textsf T}\Gamma({\boldsymbol\theta}^{(n)})]^{-1} \frac{\eps}{2} F({\boldsymbol\theta}^{(n)},{\bf v}^{(n)})\\
& = {\bf v}^{(n)} - \frac{\eps}{2} F({\boldsymbol\theta}^{(n)},{\bf v}^{(n)}) + \frac{\eps^2}{4}[{\bf I} + \frac{\eps}{2} ({\bf v}^{(n)})^{\textsf T}\Gamma({\boldsymbol\theta}^{(n)})]^{-1} [({\bf v}^{(n)})^{\textsf T}\Gamma({\boldsymbol\theta}^{(n)})] F({\boldsymbol\theta}^{(n)},{\bf v}^{(n)})\\
& = {\bf v}^{(n)} - \frac{\eps}{2} F({\boldsymbol\theta}^{(n)},{\bf v}^{(n)}) + O(\eps^2)
\end{split}
\]
Similarly, from equation \eqref{TSvsolved:2} we have
\[
{\bf v}^{(n+1)} = {\bf v}^{(n+1/2)} - \frac{\eps}{2} F({\boldsymbol\theta}^{(n+1)},{\bf v}^{(n+1/2)}) + O(\eps^2)
\]
Substituting ${\bf v}^{(n+1/2)}$ in the above equation, we obtain ${\bf v}^{(n+1)}$ as follows:
\[
\begin{split}
{\bf v}^{(n+1)} & = {\bf v}^{(n)} - \frac{\eps}{2} F({\boldsymbol\theta}^{(n)},{\bf v}^{(n)}) - \frac{\eps}{2} F({\boldsymbol\theta}^{(n+1)},{\bf v}^{(n)}) + O(\eps^2)\\
& = {\bf v}^{(n)} - F({\boldsymbol\theta}^{(n)},{\bf v}^{(n)})\eps + \frac{\eps}{2}[ F({\boldsymbol\theta}^{(n)},{\bf v}^{(n)})- F({\boldsymbol\theta}^{(n)}+O(\eps),{\bf v}^{(n)} ] + O(\eps^2)\\
& = {\bf v}^{(n)} - F({\boldsymbol\theta}^{(n)},{\bf v}^{(n)})\eps + O(\eps^2)
\end{split}
\]
From \eqref{TSvsolved:1}, \eqref{TSv:2}, and \eqref{TSvsolved:2}, we have the following numerical solution:
\[
{\bf z}^{(n+1)} = \begin{bmatrix}{\boldsymbol\theta}^{(n+1)}\\{\bf v}^{(n+1)}\end{bmatrix} = \begin{bmatrix}{\boldsymbol\theta}^{(n)}\\{\bf v}^{(n)}\end{bmatrix} +\begin{bmatrix}{\bf v}^{(n)}\\- F({\boldsymbol\theta}^{(n)},{\bf v}^{(n)})\end{bmatrix}\eps + O(\eps^2)
\]
Therefore, the local error is
\[
\begin{split}
e_{n+1} & =\Vert {\bf z}(t_{n+1}) - {\bf z}^{(n+1)}\Vert = \left\Vert \begin{bmatrix}{\boldsymbol\theta}(t_n)-{\boldsymbol\theta}^{(n)}\\{\bf v}(t_n)-{\bf v}^{(n)}\end{bmatrix} + \begin{bmatrix}{\bf v}(t_n)-{\bf v}^{(n)}\\-[F({\boldsymbol\theta}(t_n),{\bf v}(t_n))- F({\boldsymbol\theta}^{(n)},{\bf v}^{(n)})]\end{bmatrix}\eps + O(\eps^2) \right\Vert\\
& \leq (1+M\eps) e_{n} + O(\eps^2)
\end{split}
\]
where $M = c\sup_{t\in [0,T]}\Vert \nabla F({\boldsymbol\theta}(t),{\bf v}(t))\Vert$ for some constant $c>0$. Accumulating the local errors by iterating the above inequality for $L=T/\eps$ steps provides the following global error:
\[
\begin{split}
e_{n+1} & \leq (1+M\eps) e_{n} + O(\eps^2) \leq (1+M\eps)^2 e_{n-1} + 2 O(\eps^2)\leq \cdots \leq (1+M\eps)^n e_1 + n O(\eps^2)\\
& \leq (1+M\eps)^L \eps + L O(\eps^2) \leq (e^{MT}+T)\eps \to 0,\quad as\; \eps \to 0
\end{split}
\]

\subsection{Volume Correction}\label{tsvdetadj}
As before, using wedge product calculation on the system \eqref{TSvsolved:1}, \eqref{TSv:2}, and \eqref{TSvsolved:2}, the Jacobian matrix is
\[
\begin{split}
\frac{\pa ({\boldsymbol\theta}^{(n+1)},{\bf v}^{(n+1)})}{\pa ({\boldsymbol\theta}^{(n)},{\bf v}^{(n)})} = & [{\bf I}+\frac{\eps}{2} ({\bf v}^{(n+1/2)})^{\textsf T}\Gamma({\boldsymbol\theta}^{(n+1)})]^{-T} [{\bf I}-\frac{\eps}{2} ({\bf v}^{(n+1)})^{\textsf T}\Gamma({\boldsymbol\theta}^{(n+1)})]^{\textsf T} \cdot\\
& [{\bf I}+\frac{\eps}{2} ({\bf v}^{(n)})^{\textsf T}\Gamma({\boldsymbol\theta}^{(n)})]^{-T} [{\bf I}-\frac{\eps}{2} ({\bf v}^{(n+1/2)})^{\textsf T}\Gamma({\boldsymbol\theta}^{(n)})]^{\textsf T}
\end{split}
\]
As these new equations show, our derived integrator is not symplectic so the acceptance probability needs to be adjusted by the following Jacobian determinant, $\det {\bf J}$, in order to preserve the detailed balance condition:
\begin{eqnarray}\begin{array}{ll}
& \det {\bf J}_{e-LMC} := \displaystyle \left|\frac{\pa({\boldsymbol\theta}^{(L+1)},{\bf v}^{(L+1)})}{\pa({\boldsymbol\theta}^{(1)},{\bf v}^{(1)})}\right|\\
& = \displaystyle \prod_{n=1}^{L} \frac{\det(I-\eps/2 ({\bf v}^{(n+1)})^{\textsf T}\Gamma({\boldsymbol\theta}^{(n+1)})) \det(I-\eps/2 ({\bf v}^{(n+1/2)})^{\textsf T}\Gamma({\boldsymbol\theta}^{(n)}))}{\det(I+\eps/2 ({\bf v}^{(n+1/2)})^{\textsf T}\Gamma({\boldsymbol\theta}^{(n+1)})) \det(I+\eps/2 ({\bf v}^{(n)})^{\textsf T}\Gamma({\boldsymbol\theta}^{(n)}))}\\[12pt]
& = \displaystyle \prod_{n=1}^{L} \frac{\det({\bf G}({\boldsymbol\theta}^{(n+1)})-\eps/2 ({\bf v}^{(n+1)})^{\textsf T}\tilde\Gamma({\boldsymbol\theta}^{(n+1)})) \det({\bf G}({\boldsymbol\theta}^{(n)})-\eps/2 ({\bf v}^{(n+1/2)})^{\textsf T}\tilde\Gamma({\boldsymbol\theta}^{(n)}))}{\det({\bf G}({\boldsymbol\theta}^{(n+1)})+\eps/2 ({\bf v}^{(n+1/2)})^{\textsf T}\tilde\Gamma({\boldsymbol\theta}^{(n+1)})) \det({\bf G}({\boldsymbol\theta}^{(n)})+\eps/2 ({\bf v}^{(n)})^{\textsf T}\tilde\Gamma({\boldsymbol\theta}^{(n)}))}
\end{array}\label{dettsv}\end{eqnarray}
As a result, the acceptance probability is
\begin{equation*}
\alpha_{e-LMC} = \min\{1,\exp(-{\bf E}({\boldsymbol\theta}^{(L+1)},{\bf v}^{(L+1)})+{\bf E}({\boldsymbol\theta}^{(1)},{\bf v}^{(1)}))|\det {\bf J}_{e-LMC}|\}
\end{equation*}

We can now derive a completely explicit integrator for RMLMC defined in terms of $({\boldsymbol\theta},{\bf v})$. We refer to this integrator as e-RMLMC for which the corresponding steps are presented in Algorithm \ref{Alg:e-RMLMC}. In both algorithms \ref{Alg:RMLMC} and \ref{Alg:e-RMLMC}, the position update is relatively simple while the computational time is dominated by choosing the ``right'' direction (velocity) using the geometry of parameter space. Finally, it is easy to show that for ${\bf G}({\boldsymbol\theta})= {\bf I}$, our method degenerates to standard HMC.

\end{document}